\documentclass[journal,twoside,web]{ieeecolor}
\usepackage{generic}
\usepackage{cite}
\usepackage{amsmath,amssymb,amsfonts} 
\usepackage{graphicx}
\usepackage{algorithm,algorithmic}
\usepackage{hyperref}
\hypersetup{hidelinks=true}
\usepackage{textcomp}

\usepackage{subfigure}
\usepackage{tikz}  
\usepackage{enumerate} 
\usepackage{amstext}  
\usepackage{mathrsfs} 
\usepackage[abs]{overpic} 

\usetikzlibrary{calc} 
\usetikzlibrary{arrows.meta}
\usetikzlibrary{decorations.pathreplacing}

\newtheorem{theorem}{Theorem} 
\newtheorem{lemma}{Lemma} 
\newtheorem{definition}{Definition} 
\newtheorem{remark}{Remark} 
\newtheorem{corollary}{Corollary} 
 
\newtheorem{proposition}{Proposition} 
\newtheorem{example}{Example}

\def\BibTeX{{\rm B\kern-.05em{\sc i\kern-.025em b}\kern-.08em
    T\kern-.1667em\lower.7ex\hbox{E}\kern-.125emX}}
\newcommand{\proofof}[1]{\noindent\hspace{2em}{\itshape #1:}}  

\begin{document}
\title{ 
The $\theta$-Symmetric SRG with Applications to Stability of Cactus Dynamic Networks 
} 

\author{Xiaokan Yang,  
Wei Chen, \IEEEmembership{Member, IEEE}
and Li Qiu, \IEEEmembership{Fellow, IEEE}
\thanks{This work was supported in part by the National Natural Science Foundation of China under grants 62573006 and 72131001, by the Natural Science Foundation of Guangdong Province under grants 2024A1515011630 and by the Research Grants Council of Hong Kong under the General Research Fund No. 16206324.}
\thanks{Xiaokan Yang is with the School of Advanced Manufacturing and Robotics,  Peking University, Beijing 100871, China  (e-mail: yxkan21@stu.pku.edu.cn).} 
\thanks{Wei Chen is with the School of Advanced Manufacturing and Robotics and the State Key Laboratory for Turbulence and Complex
Systems, Peking University, Beijing 100871, China  (e-mail: w.chen@pku.edu.cn).}
\thanks{Li Qiu is with the School of Science and Engineering, The Chinese University of Hong Kong, Shenzhen, Guangdong, China  (e-mail: qiuli@cuhk.edu.cn).}
}

\maketitle

\begin{abstract} 
In this paper, we systematically study a variant of the scaled relative graph (SRG), referred to as the $\theta$-symmetric SRG, and apply it to the stability analysis of cactus networks. Compared with the previous SRG definition, the $\theta$-symmetric SRG enables the characterization of phase lead and lag behaviors, and serves as a more natural multivariable extension of the classical  Nyquist plot. We first analyze the gain and phase aspects of $\theta$-symmetric SRG separately and build a connection between $\theta$-segmental phase and a norm minimization problem. This connection makes it possible to compute $\theta$-segmental phase via semidefinite programming. We further derive the submultiplicative and subadditive properties of $\theta$-symmetric SRG. These algebraic properties are crucial to determine the nonsingularity of product-type and sum-type return difference matrices, which topologically correspond to the cyclic and parallel-feedback extreme cases of cactus networks.  
By taking the cyclic interconnection as the fundamental starting point, we establish necessary and sufficient conditions for its robust stability. Integrating this with the parallel case, we synthesize a unified stability framework for general multi-loop cactus networks.  
The $\theta$-symmetric SRG framework is less conservative and provides a more intuitive geometric interpretation of system behaviors compared with some existing approaches. Several examples are included to demonstrate the effectiveness of the proposed methods. 
\end{abstract}


\section{Introduction}      \label{section:introduction} 
\IEEEPARstart{G}{raphical} analysis is a fundamental method widely used in classical control theory and applications. In particular, the Nyquist plot of single-input single-output (SISO) linear time-invariant (LTI) systems, integrating both gain and phase information, provides crucial insights into the stability, robustness and performance of closed-loop systems \cite{Qiu2009IntroductionFeedbackSystems,Zhou1996Robust}. 

For multi-input multi-output (MIMO) LTI systems, the eigenloci (trajectories of eigenvalues of the transfer matrix along the Nyquist contour) play a fundamental role in the stability analysis. The generalized Nyquist criterion provides a necessary and sufficient condition for closed-loop stability in terms of eigenloci \cite{MacFarlane1977GeneralizedNyquist,Desoer1980Generalized}. However, in a feedback system composed of multiple components, the Nyquist plot of the SISO loop transfer function can be built directly from those of its components, whereas the eigenloci of the MIMO return ratio bear no explicit connection to the eigenloci of each component. In other words, the eigenloci of the loop transfer matrix cannot be inferred from the behavior of individual components. Moreover, for large-dimensional transfer matrices, plotting the eigenloci becomes numerically challenging.

A classical method to overcome the aforementioned issues is to employ graphical over-approximation of the eigenloci, enclosing all eigenloci as tightly as possible while maintaining desirable analytical properties. Specifically, the well-known gain-based $\mathcal{H}_\infty$ theory offers a graphical method by bounding the gain  
\cite{Zhou1996Robust}.   
In parallel with the gain, considerable effort has been devoted to developing the phase counterpart recently. 
Some notable developments include  
sectorial phases \cite{Chen2024Phase,Tits1999Robustness}, singular angles \cite{Chen2025SingularAngle,Wielandt1967Topics,Gustafson1968Angle,Krein1969AngularLocalization,Gustafson1994Antieigenvalues} and segmental phases \cite{Chen2026Cyclic}. 
Additionally, the theories of positive real systems \cite{Anderson1973Network} and negative imaginary systems \cite{Petersen2010Feedback_control_NI} are closely related, as both incorporate qualitative phase information of the systems. 

Combining gain and phase together offers a more delicate framework than treating them separately. When considering the interconnection of two subsystems, some mixed gain-phase graph separation methods work well for stability analysis \cite{Li2008Eigenvalues,Liang2024Feedback,Zhang2025DW_Shell}. 
However, when applied to multi-loop networks, these graph separation methods are no longer applicable. A typical multi-loop architecture is the cactus network, which means that each pair of distinct simple loops has at most one common vertex. Two extreme cases of cactus networks are the single-loop cyclic interconnection and the parallel interconnection of multiple loops. The cactus networks arise in various application domains such as biochemical networks \cite{Tyson1978Dynamics,Hori2011CyclicGene}, large-scale systems  \cite{Arcak2006DiagonalStability,Sontag2006SecantCondition,Pates2023Generalisation,Liu2011LyapunovISS_Cyclic_small_gain} and power systems \cite{Monshizadeh2019Secant}. There is therefore a clear need for a mixed gain-phase geometric tool that can effectively handle these network architectures. This is a motivation of the $\theta$-symmetric scaled relative graph (SRG) framework in this paper, which is a variant of the SRG.

In particular, the SRG, originally proposed in the optimization literature \cite{Hannah2016ScaledRelativeGraph,Ryu2022SRG} and subsequently brought into control community by \cite{Pates2021Scaled,Chaffey2023Graphical}, has attracted considerable attention and inspired many follow-up works \cite{vandenEijnden2025Phase,Chen2025Graphical,Krebbekx2025Graphical,deGroot2025DissipativitySRG,Baron2025MixedGainPhase}.  
However, the current definition of SRG encodes only the magnitude information of the phase but fails to retain its sign information. Consequently, the SRG of a complex scalar matrix appears as a conjugate pair. This limitation of the SRG definition stems from a more fundamental problem on how to appropriately define the angle between two complex vectors. To overcome this, the $\theta$-symmetric SRG introduced in this paper builds on a more general definition of complex vector angles, effectively capturing both the magnitude and sign information of the phase. 

In this paper, we systematically study the notion of $\theta$-symmetric SRG, including its implications for gain-phase characterization, and develop a unified framework for the stability analysis of cactus networks. The main contributions of this paper are summarized as follows. 
\begin{enumerate}[1)]
    \item We develop a systematic formulation of the $\theta$-symmetric SRG and investigate its properties. In particular, we establish its submultiplicative and subadditive properties, which form the foundation for the stability analysis. 
    \item We reveal a connection between the $\theta$-segmental phase and a norm minimization problem by analyzing the phase aspect of $\theta$-symmetric SRG. This connection immediately leads to a semidefinite programming (SDP) approach for computing the $\theta$-segmental phase, thereby resolving an open problem noted in \cite{Chen2025SingularAngle}.  
    \item We derive nonsingularity criteria for both product-type and sum-type return difference matrices via the $\theta$-symmetric SRG. These criteria enjoy clear graph interpretations and subsume several existing results as special cases. 
    \item We establish necessary and sufficient conditions for the robust stability of cyclic interconnections, which serve as the fundamental building blocks of cactus networks. 
    The sufficiency follows from a frequency-wise application of the matrix results, whereas the necessity relies on a constructive interpolation method. 
    \item We synthesize a unified stability framework for cactus dynamic networks via the $\theta$-symmetric SRG. A hierarchical iterative procedure is designed to determine stability in a layer-by-layer manner. 
\end{enumerate}

The rest of the paper is organized as follows. Section \ref{section:preliminaries} provides a brief review of the background and basic concepts of the SRG. In Section \ref{section:matrix_theta_SRG}, we introduce the notion of $\theta$-symmetric SRG along with its properties.  
Section \ref{section:segmental_phase_computation} focuses on the phase of $\theta$-symmetric SRG, establishing a connection between $\theta$-segmental phase and a norm minimization problem. 
Section \ref{section:subadd_submultiply} studies the submultiplicative and subadditive properties of $\theta$-symmetric SRG and derives nonsingularity criteria for return difference matrices. Section \ref{section:robust_nonsingularity_matrices} develops necessary and sufficient conditions for nonsingularity under mixed gain-phase  uncertainties. Section \ref{section:single_loop_stability_systems} defines the frequency-wise $\theta$-symmetric SRG of a system and establishes necessary and sufficient conditions for robust stability of cyclic interconnected systems. Section \ref{section:cactus_graph} synthesizes a unified stability framework for general cactus dynamic networks. The paper is concluded in Section \ref{section:conclusion}.  
Preliminary results of Section \ref{section:matrix_theta_SRG}, \ref{section:robust_nonsingularity_matrices} and \ref{section:single_loop_stability_systems} appeared in the conference papers \cite{Yang2026ThetaSRG,Yang2026ThetaSRGFurther}. 

The notation used in this paper is mostly standard. Let $\mathbb{R}$ and $\mathbb{C}$ be the sets of real and complex scalars. Let $\mathbb{R}^{+}$ and $\mathbb{R}^{-}$ be the sets of nonnegative and nonpositive real numbers, respectively. The transpose, conjugate and conjugate transpose are denoted by $(\cdot)^T, \overline{(\cdot)}$ and $(\cdot)^H$, respectively. The inner product between two vectors $x,y\in\mathbb{C}^{n}$ is denoted by $\langle x, y\rangle = x^Hy$.  
Let $\|\cdot\|$ denote the Euclidean norm of a vector and the spectral norm of a matrix. 
The spectrum of a matrix $C \in \mathbb{C}^{n \times n}$ is denoted by $\Lambda(C)$. Denote the real part and imaginary part of a matrix $C$ as $\mathrm{Re} \,C = \frac{1}{2}(C+C^H)$ and $\mathrm{Im}\, C = \frac{1}{2j}(C - C^H)$, respectively. The largest and smallest singular values of $C$ are denoted by $\overline{\sigma}(C)$ and $\underline{\sigma}(C)$, respectively. Let $\mathcal{RH}_\infty^{m\times m}$ denote the set of $m \times m$ real rational proper stable transfer matrices.

\section{Preliminaries}             \label{section:preliminaries}

 The notion of the SRG was first defined for operators on real Hilbert spaces in \cite{Hannah2016ScaledRelativeGraph}, and  served as a new geometric tool to analyze contractive and non-expansive fixed-point iterations. Specializing to the finite-dimensional matrix case, the SRG of a real matrix $R \in \mathbb{R}^{n\times n}$ is defined as \cite{Huang2019SRG_Normal} 
\begin{equation*}
    \mathcal{G}(R) = \left\{  \frac{\|Rx\|}{\| x \|} \exp \left\{ \pm j \arccos \frac{\langle x,y \rangle}{\| x \| \|y\|} \right\}: 0\neq x \in \mathbb{R}^{n} \right\}, 
\end{equation*}
which is a subset of $\mathbb{C}$ and symmetric with respect to the real axis. The SRG defined in this way for real matrices is reasonable and natural, given that the spectrum of a real matrix is symmetric with respect to the real axis.  

When the concept of SRG was introduced into the control theory in \cite{Pates2021Scaled,Chaffey2023Graphical}, its definition was extended from operators on real Hilbert spaces to those on complex Hilbert spaces. In these works, the SRG is defined in terms of the angle between complex vectors. 
For vectors $x,y\in\mathbb{C}^{n}$, the angle $\angle (x,y) \in [0,\pi]$ between $x$ and $y$ is defined by \cite{Wielandt1967Topics,Gustafson1997NumericalRange} 
\begin{equation}                        \label{eq:vector_angle}
    \angle(x,y) = \arccos \frac{\mathrm{Re}\,\langle x,y \rangle}{\| x \| \|y\|}, \quad \text{ if } \  x \neq 0, y \neq 0 ,    
\end{equation} and $\angle (x,y) = 0$ if either $x=0$ or $y=0$. Then the SRG of a complex matrix $C \in \mathbb{C}^{n\times n}$ is defined as   
\begin{equation}                    \label{eq:SRG_matrix} 
    \mathrm{SRG}(C) = \left\{  \frac{\|Cx \|}{\| x\|} \exp\{ {\pm j\angle(x,Cx)} \}: 0 \neq x \in \mathbb{C}^n \right\} ,  
\end{equation}
which is also symmetric with respect to the real axis. For consistency with our terminology, we refer to \eqref{eq:SRG_matrix} as the ``$0$-symmetric SRG''. The $0$-symmetric SRG inherently makes $\mathrm{SRG}(C)$ coincide with $\mathrm{SRG}(C^H)$ and leads to the SRG of a complex scalar matrix appearing as a conjugate pair. An illustrated example is shown below. 

\begin{example}                 \label{exmp:simple_SRG_example}
    Consider the following cases: 
    \begin{enumerate}[(i)]
        \item Let $ C_1 \!  = \! j,   C_2\!  = \! -j,   C_3 \! = \!  
    \begin{bmatrix}    j & 0 \\ 0  & -j   \end{bmatrix} . $ 
    Then 
    \begin{equation*}
        \mathrm{SRG}(C_1)=\mathrm{SRG}(C_2)=\mathrm{SRG}(C_3)=\{ j, -j \} , 
    \end{equation*} 
    which is a conjugate pair. 
    \item Let $C = \begin{bmatrix}
        j & 0 \\ 0 & 1 
    \end{bmatrix}.$ Then 
    \begin{equation*}
        \mathrm{SRG}(C) = \mathrm{SRG}(C^H) = \left\{ e^{j\phi} \!\in\! \mathbb{C}: \phi \!\in\! \left[-\frac{\pi}{2}, \frac{\pi}{2}\right] \right\}. 
    \end{equation*} 
    \end{enumerate} 
\end{example}

A crucial observation is that the $0$-symmetric SRG does not commute with the complex rotation in general, i.e., 
\begin{equation}                    \label{eq:SRG_rotation_not_commute}
    e^{-j\theta} \mathrm{SRG}(C) \neq \mathrm{SRG}(e^{-j\theta}C) . 
\end{equation} 
This is contrast to the numerical range \cite{Gustafson1997NumericalRange}, defined as 
\begin{equation}        \label{eq:numerical_range} 
    W(C) = \{ x^HCx  : x \in \mathbb{C}^n, \|x\|=1 \} . 
\end{equation} 
which satisfies $e^{-j\theta}W(C) = W(e^{-j\theta} C)$. 

For a matrix set $\mathcal{C} \subset \mathbb{C}^{n\times n}$, the $0$-symmetric SRG of $\mathcal{C}$ is defined by  $\mathrm{SRG}(\mathcal{C}) = \bigcup_{C \in \mathcal{C}} \mathrm{SRG}(C) $. Note that a matrix set can be a single matrix.

\begin{figure}[htb] 
\vspace{-10pt}
\centering 
\subfigure[]{
\begin{minipage}[t]{0.319\linewidth}
\centering 
\begin{tikzpicture} 
\draw[-latex, thick](0,2.1)--(2.5,2.1); 
\draw[-latex, thick](0.5,0.7)--(0.5,3.7); 
\node[below] at (2.3,2.05){$\mathrm{Re}$}; 
\node[left] at(0.45,3.6){$\mathrm{Im}$}; 
\node[left] at(0.5,1.9){$0$}; 
\coordinate (Z) at (1.2, 2.9); 
\coordinate (Zconj) at (1.2, 1.3);   
\draw[very thick] (Z) -- (Zconj); 
\filldraw (Z) circle (1.5pt) node[above right] {$z$};
\filldraw (Zconj) circle (1.5pt) node[below right] {$\overline{z}$};
\end{tikzpicture}  
\label{fig:chord}
\end{minipage}%
}%
\hfill%
\subfigure[]{
\begin{minipage}[t]{0.304\linewidth}
\centering 
\begin{tikzpicture} 
\def\xx{0.8} 
\def\yy{2.1} 
\draw[-latex, thick](0.2,\yy)--(2.6,\yy); 
\draw[-latex, thick](\xx,0.7)--(\xx,3.7); 
\node[below] at (2.45,2.05){$\mathrm{Re}$}; 
\node[left] at(\xx-0.05,3.6){$\mathrm{Im}$}; 
\node[left] at(\xx,\yy-0.2){$0$};  
\def\x{1.4} 
\def\y{2.9}
\coordinate (Z) at (\x, \y); 
\coordinate (Zconj) at (\x, 2*\yy-\y);   
\pgfmathsetmacro{\R}{sqrt((\x-\xx)*(\x-\xx) + (\y-\yy)*(\y-\yy))}
\pgfmathsetmacro{\Angle}{atan2(\y-\yy, \x-\xx)} 
\draw[very thick] (Z) arc (\Angle:-\Angle:\R);  
\draw[dashed] (\xx,\yy) -- (Z);
\draw[dashed] (\xx,\yy) -- (Zconj);
\filldraw (Z) circle (1.5pt) node[above right] {$z$};
\filldraw (Zconj) circle (1.5pt) node[below right] {$\overline{z}$}; 
\end{tikzpicture} 
\label{fig:arc_plus}
\end{minipage}%
}%
\hfill%
\subfigure[]{
\begin{minipage}[t]{0.292\linewidth}
\centering 
\begin{tikzpicture}
\def\xx{1.5} 
\def\yy{2.1} 
\draw[-latex, thick](0.2,\yy)--(2.6,\yy); 
\draw[-latex, thick](\xx,0.7)--(\xx,3.7); 
\node[below] at (2.45,2.05){$\mathrm{Re}$}; 
\node[left] at(\xx-0.05,3.6){$\mathrm{Im}$}; 
\node[left] at(\xx,\yy-0.2){$0$};  
\def\x{2.1} 
\def\y{2.9}
\coordinate (Z) at (\x, \y); 
\coordinate (Zconj) at (\x, 2*\yy-\y);   
\pgfmathsetmacro{\R}{sqrt((\x-\xx)*(\x-\xx) + (\y-\yy)*(\y-\yy))}
\pgfmathsetmacro{\Angle}{atan2(\y-\yy, \x-\xx)} 
\draw[very thick] (Z) arc (\Angle:360-\Angle:\R);  
\draw[dashed] (\xx,\yy) -- (Z);
\draw[dashed] (\xx,\yy) -- (Zconj);
\filldraw (Z) circle (1.5pt) node[above right] {$z$};
\filldraw (Zconj) circle (1.5pt) node[below right] {$\overline{z}$}; 
\end{tikzpicture} 
\label{fig:arc_minus}
\end{minipage} 
} 
\vspace{-0.5em}  
\caption{Illustrations of (a)  {\unboldmath{$\mathrm{Chord}(z)$}}, (b) {\unboldmath{$\mathrm{Arc}^+(z)$}} and (c) {\unboldmath{$\mathrm{Arc}^-(z)$}}. }    
\label{fig:different_phase_cases} 
\vspace{-0.5em} 
\end{figure}
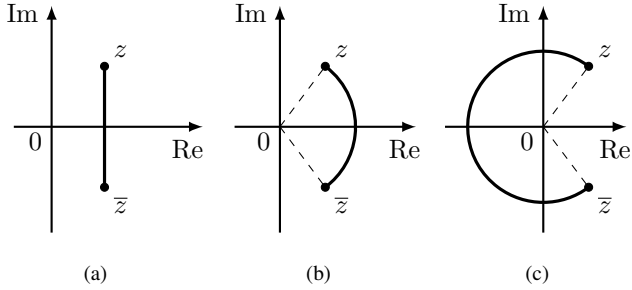

Let $\mathrm{Chord}(z)$ denote the line segment connecting point $z \in \mathbb{C}$ and its complex conjugate $\overline{z}$, as shown in Fig.~\ref{fig:chord}. A matrix set $\mathcal{C} \subset \mathbb{C}^{n\times n}$ is said to satisfy the chord property if $\mathrm{Chord}(z) \subset \mathrm{SRG}(\mathcal{C})$ for all $z \in \mathrm{SRG}(\mathcal{C})$. 
Similarly, let $\mathrm{Arc}^+(z)$ and $\mathrm{Arc}^-(z)$ denote the circular arcs centered at the origin connecting $z$ and $\overline{z}$ that cross the positive and negative real semi-axes, respectively, as shown in Fig.~\ref{fig:arc_plus} and Fig.~\ref{fig:arc_minus}. Then $\mathcal{C}$ is said to satisfy the arc property if either $\mathrm{Arc}^+(z) \subset \mathrm{SRG}(\mathcal{C}) $ for all $z \in \mathrm{SRG} (\mathcal{C})$, or $\mathrm{Arc}^-(z) \subset \mathrm{SRG}(\mathcal{C}) $ for all $z \in \mathrm{SRG} (\mathcal{C})$.

\begin{lemma}             
[\!\! \cite{Ryu2022SRG,Chaffey2023Graphical}]  \label{lem:SRG_chord_arc_property} 
    Let matrix sets $\mathcal{A},\mathcal{B} \subset \mathbb{C}^{n\times n}$ be given, then 
    \begin{enumerate}[(i)]
        \item if either $\mathcal{A}$ or $\mathcal{B}$ satisfies the chord property, then 
        \begin{equation*}
            \mathrm{SRG}(\mathcal{A} + \mathcal{B}) \subset \mathrm{SRG}(\mathcal{A}) + \mathrm{SRG}(\mathcal{B}) ; 
        \end{equation*} 
        \item if either $\mathcal{A}$ or $\mathcal{B}$ satisfies the arc property, then 
        \begin{equation*}
            \mathrm{SRG}(\mathcal{A} \mathcal{B}) \subset \mathrm{SRG}(\mathcal{A}) \mathrm{SRG}(\mathcal{B}). 
        \end{equation*}
    \end{enumerate}
\end{lemma}

Indeed, many typical uncertainty sets, including norm-bounded matrices and positive real matrices, naturally satisfy both chord and arc properties. This indicates that the chord and arc properties of matrix sets above are not ad hoc but rather reflect some geometric structure of the uncertainty sets.

\begin{remark}
    In general, a single matrix does not satisfy the arc or chord property. However, some special matrices do. For example, a real scalar matrix satisfies both arc and chord property. A unitary matrix with at least one real eigenvalue satisfies arc property. 
\end{remark}

\section{$\theta$-Symmetric SRGs of Complex Matrices}          \label{section:matrix_theta_SRG}  

As discussed in Section \ref{section:preliminaries}, the current SRG formulation exhibits certain limitations that appear to stem from a more fundamental issue, namely the definition of the angle between complex vectors. The definition in \eqref{eq:vector_angle} is classical yet relies exclusively on the real part of the complex inner product.  
Nevertheless, in complex spaces, the real part is not intrinsically special, and using it to define angles effectively privileges the real axis as a reference direction---an arbitrary choice lacking geometric justification. For instance, consider any nonzero $x \in \mathbb{C}^n$ and let $y = jx$. According to \eqref{eq:vector_angle}, the angle between $x$ and $y$ is $\frac{\pi}{2}$, whereas using the imaginary part instead yields the angle is $0$. 

Motivated by this, we begin with a more general definition of the angle between complex vectors. Then we adopt this definition to formulate the $\theta$-symmetric SRG and derive its properties and implications. 

For given vectors $x, y \in \mathbb{C}^{n}$ and $\theta \in \mathbb{R}$, define $\angle_\theta(x,y) \in [0,\pi]$ as 
\begin{equation}                \label{eq:vector_angle_phase_related}
    \angle_\theta(x,y) = \arccos \frac{\mathrm{Re}\langle x, e^{-j\theta} y \rangle}{\|x\|\|y\|} \quad \text{ if } \  x \neq 0, y \neq 0 ,    
\end{equation} 
and $\angle_\theta (x,y) = 0$ if either $x=0$ or $y=0$.  
Based on \eqref{eq:vector_angle_phase_related}, the $\theta$-symmetric SRG of a matrix $C \in \mathbb{C}^{n\times n}$ is defined as 
\begin{equation}                \label{eq:definition_PS_SRG_matrix}
    \mathrm{SRG}_\theta(C) \! = \! \left\{ \! 
        \frac{\|  Cx \| }{\| x \|} \exp\{{j(\theta \! \pm \! \angle_\theta(x, Cx)}\} \! :   0 \! \neq \! x \! \in \! \mathbb{C}^n \! 
    \right\} . 
\end{equation} 
In view of \eqref{eq:SRG_matrix}, an explicit connection between $\theta$-symmetric SRG and $0$-symmetric SRG is 
\begin{equation}                            \label{eq:definition_theta_SRG_matrix} 
    \mathrm{SRG}_\theta (C) = e^{j\theta} \mathrm{SRG}(e^{-j\theta} C) . 
\end{equation} 
By \eqref{eq:SRG_rotation_not_commute}, we mention that $\mathrm{SRG}_\theta(C)$ is not obtained by a simple rotation of $\mathrm{SRG}(C)$ and that they can differ substantially in shape.   
Clearly, when $\theta = 0$, \eqref{eq:definition_PS_SRG_matrix} reduces to the $0$-symmetric SRG defined in \eqref{eq:SRG_matrix}. For any $\theta\in\mathbb{R}$, $\mathrm{SRG}_\theta(C)$ is a subset of $\mathbb{C}$ and is symmetric with respect to the line $e^{j\theta} \mathbb{R}$. This property motivates the terminology ``$\theta$-symmetric SRG" that we adopt here. 

\begin{figure}[htb]
    \centering
    \begin{overpic}[width=0.5\linewidth]{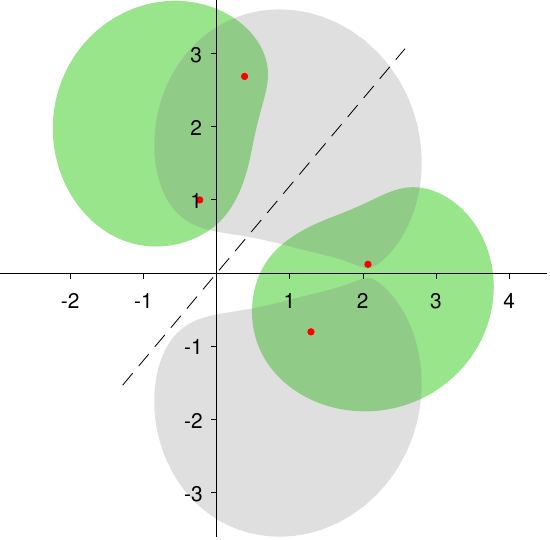}
        \put(61,63){\footnotesize $\theta$}   
        \put(118,54){ \footnotesize $\mathrm{Re}$} 
        \put(38,118){\footnotesize $\mathrm{Im}$} 
        \begin{tikzpicture}[overlay] 
            \draw[->] (2.1,2.15) arc(0:48.5:0.35);  
        \end{tikzpicture} 
    \end{overpic}
    \caption{An illustration of {\unboldmath$\mathrm{SRG}(C)$} (gray region) and {\unboldmath$\mathrm{SRG}_\theta(C)$} (green region) with {\unboldmath$\theta = 50^\circ$}, where the red points are the eigenvalues of {\unboldmath$C$}. }
    \label{fig:SRG_theta_SRG} 
\end{figure}
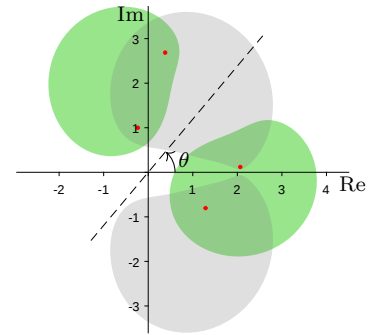

\begin{example}
    Consider a matrix 
    \begin{equation*}
        C = \begin{bmatrix}
            1+2j & 0 & 2 & 0 \\ 
            0 & 1+j & -1 & 1 \\ 
            -j & 1 & 1.5 & 0 \\ 
            2 & -j & 1 & 0 
        \end{bmatrix}. 
    \end{equation*} Let $\theta = 50^\circ$, then an illustration of $\mathrm{SRG}(C)$ and $\mathrm{SRG}_\theta(C)$ is shown in Fig.~\ref{fig:SRG_theta_SRG}. 
\end{example}

The $\theta$-symmetric SRG of a matrix set can be defined similarly. For a matrix set $\mathcal{C} \subset \mathbb{C}^{n\times n}$, the $\theta$-symmetric SRG of $\mathcal{C}$ is defined by 
$ \mathrm{SRG}_\theta (\mathcal{C}) = \bigcup_{C\in\mathcal{C}} \mathrm{SRG}_\theta(C)$. 

Some properties of $\theta$-symmetric SRG are summarized as follows. The proofs are relatively straightforward and can be found in \cite{Yang2025Cascade_theta_SRG}. 

\begin{lemma}             \label{lem:properties_theta_SRG} 
    Let matrix $C\in \mathbb{C}^{n\times n}$ and $\theta \in \mathbb{R}$  be given, then  
    \begin{enumerate}[(i)] 
        \item \label{item:theta_SRG_property_1}  spectrum containment $\Lambda(C) \subset \mathrm{SRG}_\theta(C)$ holds; 
        \item \label{item:theta_SRG_property_2}  for any $0 \neq \mu \in \mathbb{R}$, we have 
        $ \mathrm{SRG}_\theta(\mu C) = \mu \mathrm{SRG}_\theta(C) $ and 
        $ \mathrm{SRG}_\theta (C + \mu e^{j\theta} I) = \mathrm{SRG}_\theta(C) + \mu e^{j\theta} $;  
        \item \label{item:theta_SRG_property_3} we have 
        $\mathrm{SRG}_{\theta+\pi} (C) = \mathrm{SRG}_\theta (C) $;  
        \item \label{item:theta_SRG_property_4}  if $C$ is nonsingular, then 
        \begin{equation*}
            \mathrm{SRG}_\theta(C^{-1}) = ( \mathrm{SRG}_{-\theta} (C) )^{-1} . 
        \end{equation*} 
    \end{enumerate}
\end{lemma} 

For given $\theta \in \mathbb{R}$, the phase spread of the $\theta$-symmetric SRG is defined as 
\begin{equation}             \label{eq:phase_spread}
    \Gamma_\theta(C) = \sup_{x \in \mathbb{C}^n} \,  \angle_\theta (x,Cx) . 
\end{equation} 
Recall that $C$ is called sectorial if $0\notin W(C)$ \cite{Chen2024Phase}.  
Interestingly, sectorial matrices can be equivalently characterized within the $\theta$-symmetric SRG framework: a matrix $C$ is sectorial if and only if there exists $\theta \in \mathbb{R}$ such that $\Gamma_\theta(C) < \frac{\pi}{2}$.

When restricted to sectorial matrices, we can define the unique optimal $\theta^\star \in \mathbb{R}$ to minimize the phase spread of the $\theta$-symmetric SRG as \cite{Chen2026Cyclic} 
\begin{equation}                        \label{eq:optimal_theta}  
    \theta^\star = \arg\min_{\theta\in\mathbb{R}} \,  \Gamma_\theta(C) .  
\end{equation} 
Accordingly, $\theta^\star$-symmetric SRG is referred to as the canonical SRG, denoted by $\mathrm{SRG}_{\theta^\star}(C)$.

The following example demonstrates that the canonical SRG yields a tighter graphical enclosure of the eigenvalues than the $0$-symmetric SRG, thereby effectively reducing the conservatism.

\begin{example}             \label{examp:for_optimal_theta_SRG} 
     Consider the following sets of matrices: 
    \begin{enumerate}[(i)]
        \item Let $C \!= \! z I $ be a complex scalar matrix with $z \!\in\! \mathbb{C}$. Then $\mathrm{SRG}_{\theta^\star}(C) \! =\!  z$ with $\theta^\star \! = \! \angle z$, while $\mathrm{SRG}(C) \! = \! \{z,\overline{z}\}$.  
        \item Let $C\in \mathbb{C}^{n\times n}$ be a unitary matrix with $\Lambda(C) = \{ e^{j\phi_1}, \dots, e^{j\phi_n} \}, \, \phi_1\geq \cdots \geq \phi_n,$ satisfying $0 \leq \phi_1 - \phi_n < \pi$. Then 
        \begin{equation*}
            \mathrm{SRG}_{\theta^\star}(C) \! =\! \{ e^{j\phi}\!\in\!\mathbb{C}: \delta_{\min} \leq |\phi-\theta^\star| \leq \delta_{\max}  \} , 
        \end{equation*}  
        with $\theta^\star = (\phi_1+\phi_n)/2 , 
        \delta_{\max} = \max_{1 \leq i \leq n  } |\phi_i-\theta^\star| $ and $ \delta_{\min} = \min_{1\leq i \leq n} | \phi_i - \theta^\star | $.   
        The $0$-symmetric SRG of $C$ is 
        \begin{equation*}
            \mathrm{SRG}(C) \!=\! \{ e^{j\phi} \!\in\! \mathbb{C}\!:  \phi_{\min} \leq |\phi|\leq\phi_{\max}  \}  , 
        \end{equation*}   
        with $ \phi_{\max} \!=\! \max_{1 \leq i \leq n } |\phi_i| $ and $  \phi_{\min} \!=\! \min_{1 \leq i \leq n } |\phi_i|  $.  
         \item Let $C \in \mathbb{C}^{n\times n}$ be a positive definite matrix. Then $\theta^\star = 0$ and $\mathrm{SRG}_{\theta^\star} (C) = \mathrm{SRG}(C)$.  
    \end{enumerate}
\end{example}
 
\section{Gain and Phase of $\theta$-Symmetric SRG}
\label{section:segmental_phase_computation} 

The $\theta$-symmetric SRG is a mixed gain-phase geometric representation, allowing its gain and phase aspects to be considered separately. In view of \eqref{eq:definition_PS_SRG_matrix}, the gain aspect corresponds to the maximum singular value, a classical and well-studied quantity. We therefore focus on the phase aspect, which corresponds to the recently proposed notion of $\theta$-segmental phase \cite{Chen2026Cyclic}. For a matrix $C \in \mathbb{C}^{n\times n}$ and $\theta \in \mathbb{R}$, the $\theta$-segmental phase $\Psi_\theta(C)$ is defined as  
\begin{equation*}
    \Psi_\theta(C) = \left[ \underline{\psi}_\theta(C), \overline{\psi}_\theta(C) \right] , 
\end{equation*} 
where 
\begin{equation*}
    \overline{\psi}_\theta (C) = \theta + \Gamma_\theta(C), \quad \underline{\psi}_\theta(C) = \theta - \Gamma_\theta(C) , 
\end{equation*}
with $\Gamma_\theta(C)$ given in \eqref{eq:phase_spread}. When $\theta = 0$, $\Gamma_0(C)$ is called the singular angle of $C$ \cite{Wielandt1967Topics} and 
\begin{equation*}
    \Psi_0(C) = [-\Gamma_0(C), \Gamma_0(C)] . 
\end{equation*}
When $\theta^\star$ is determined by \eqref{eq:optimal_theta}, the corresponding $\Psi_{\theta^\star}(C)$ defines the segmental phase of $C$, which attains the minimal phase spread.  

\begin{remark}
    The notion of $\theta$-segmental phase was originally introduced based on the notion of normalized numerical range \cite{Chen2026Cyclic,Chen2025SingularAngle,Auzinger2003Sectorial,Lins2018NormalizedNR}.  Here, we take a new perspective and interpret it as the phase aspect of the $\theta$-symmetric SRG. 
\end{remark}

The following result establishes a connection between the $\theta$-segmental phase and a norm minimization problem.  

\begin{theorem}             \label{thm:segmental_phase_norm_minimization} 
    Let matrix $C \in \mathbb{C}^{n\times n}$ be sectorial, then  
    \begin{enumerate}[(i)]
        \item for any $\theta \in \mathbb{R}$ such that $\Gamma_\theta(C) < \frac{\pi}{2}$, there holds  
        $  \Psi_\theta (C) = \theta + [-\Gamma_\theta, \Gamma_\theta], $  
        where  
        \begin{equation*}
            \Gamma_\theta = \arcsin \, \min_{\gamma>0} \|\gamma e^{-j\theta} C-I\| . 
        \end{equation*} 
        \item there holds  
         $  \Psi_{\theta^\star} (C) = \theta^\star + [-\Gamma^\star, \Gamma^\star], $ 
        where  
        \begin{align*}
            \Gamma^\star & = \arcsin \min_{z \in \mathbb{C}}\|zC-I\|,  \\ 
            \theta^\star & = - \angle \arg \min_{z \in \mathbb{C}} \|zC-I\| . 
        \end{align*} 
    \end{enumerate}
\end{theorem}

\vspace{3pt}

Before proving Theorem \ref{thm:segmental_phase_norm_minimization}, we first introduce the following result, which provides a key technical ingredient for the proof. 

\begin{proposition}                 \label{prop:theta_segmental_phase_LMI_graphical_interpretation}
    Let matrix $C\in\mathbb{C}^{n\times n}$ and $\theta\in\mathbb{R}, \delta \in [0,\frac{\pi}{2})$ be given, then the following statements are equivalent. 
    \begin{enumerate}[(i)]
        \item There holds $\Psi_\theta(C) \subset \theta + [-\delta, \delta] $. 
        \item There exists $\mu>0$ such that 
        \begin{equation*}
            W(T) \subset \{ a+bj \in \mathbb{C}:  2\mu a \geq \cos\delta (b+\mu^2), a , b \in \mathbb{R} \} , 
        \end{equation*} where $T = \mathrm{Re}\{ e^{-j\theta} C\} + jC^HC $. 
        
        \item There exists $\mu>0$ such that 
        \begin{equation}             \label{eq:LMI_theta_segmental_phase}  
        \begin{bmatrix}
            C^HC - \frac{\mu}{\cos\delta} (e^{-j\theta}C + e^{j\theta}C) & \mu I \\ \mu I  & - I 
        \end{bmatrix} \leq 0 . 
        \end{equation} 
    \item There exists $\mu>0$ such that 
    \begin{equation}            \label{eq:SRG_disk}
        \mathrm{SRG}_\theta (C) \subset \mathcal{D}[\theta,\mu] , 
    \end{equation}  
    where 
    \begin{align*}
    \mathcal{D}[\theta,\mu] & = \{ e^{j\theta} (a+bj) \in \mathbb{C}: \\ 
    & \quad \ \  \left( a - \frac{\mu}{\cos\delta}  \right)^2 +  b^2 \leq \mu^2\tan^2\delta, \ a ,b \in \mathbb{R} \} . 
    \end{align*}
    \end{enumerate} 
\end{proposition}

\vspace{1em} 

\begin{proof}
    See Appendix \ref{app:proofs}. 
\end{proof}

For a fixed $\theta \in \mathbb{R}$, Proposition \ref{prop:theta_segmental_phase_LMI_graphical_interpretation} establishes an equivalent LMI characterization of the $\theta$-segmental phase condition and establishes equivalent graphical interpretations in terms of the numerical range and $\theta$-symmetric SRG, respectively. 

An illustration of these graphical interpretations is shown in Fig.~\ref{fig:W_T_numerical_range_SRG_disk}. One can see from Fig.~\ref{subfig:SRG_disk} that a matrix $C$ satisfies $\Psi_\theta(C) \subset \theta + [-\delta,\delta]$ if and only if there exists a parameter $\mu>0$ such that $\mathrm{SRG}_\theta(C)$ is contained in a disk centered at $\frac{\mu }{\cos\delta} e^{j\theta}$ with radius $\mu\tan\delta$. 

\begin{figure}[htb] 
\centering 
\subfigure[]{
\begin{minipage}[t]{0.46\linewidth}
\centering 
\begin{tikzpicture}
\begin{scope}  
\draw[white, fill=gray, opacity=0.5, rotate around={65:(2.67,2.37)}]
(2.67,2.37) ellipse (1 and 0.4);   
\end{scope} 
\draw[-latex, thick](0.3,1.2)--(3.8,1.2); 
\draw[-latex, thick](1.2,0.2)--(1.2,4); 
\draw[domain=0.4:2.7, smooth, variable=\x] 
plot ({\x},{1.2+1.1*(\x-1.2)*(\x-1.2)});
\coordinate (z0) at (2.27,{1.1*1.07*1.07+1.2}); 
\draw[blue] ($(z0)+(-112.5:2)$) -- ($(z0)+(67.5:1.2)$);
\fill[red] (z0) circle(0.8pt);   
\node[below] at (3.6,1.2){$a$}; 
\node[left] at(1.2,3.8){$b$}; 
\node[left] at(1.2,1){$0$};   
\node[right] at(2.15,2.17){{\footnotesize $W\!( T )\!$}};  
\end{tikzpicture} 
\label{subfig:W_T_numerical_range} 
\end{minipage} 
}
\subfigure[]{
\begin{minipage}[t]{0.46\linewidth}
\centering 
\begin{tikzpicture}
\begin{scope}   
\draw [white, fill = gray, opacity = 0.5] ($(0.8,1.2)+(20:1.6)$) circle(1); 
\draw (0.8,1.2) -- ++(58.7:2.6);
\draw (0.8,1.2) -- ++(-18.7:2.6);  
\draw [blue] (0.8,1.2) -- ++(20:3); 
\draw[->] ($(0.8,1.2)+(0:0.7)$) arc(0:20:0.7);  
\draw[->] ($(0.8,1.2)+(20:0.55)$) arc(20:58:0.55);  
\end{scope} 
\draw[-latex, thick](0.3,1.2)--(4,1.2); 
\draw[-latex, thick](0.8,0.2)--(0.8,4);  
\node[below] at (3.8,1.2){$\mathrm{Re}$}; 
\node[left] at(0.8,3.8){$\mathrm{Im}$}; 
\node[left] at(0.8,1.4){$0$}; 
\node[right] at (1.2,1.7){$\delta$};  
\node[below] at (1.75,1.62){$\theta$};  
\node[below] at (2.35,2.6){$\mathcal{D}[\theta,\mu]$}; 
\draw[decorate, decoration={brace, amplitude=6pt, mirror}]
(0.8,1.2) -- ($(0.8,1.2)+(-18.7:1.25)$) 
node[midway, sloped, transform shape, yshift=-12pt] {$\mu$}; 
\draw [dashed] ($(0.8,1.2)+(20:1.6)$) -- ($(0.8,1.2)+(-18.7:1.25)$); 
\draw [red, fill=red] ($(0.8,1.2)+(20:1.6)$) circle(0.8pt);  
\end{tikzpicture} 
\label{subfig:SRG_disk}  
\end{minipage} 
} 
\vspace{-0.5em}  
\caption{(a) An illustrations of {\unboldmath$W(T)$} (gray region), where the parabola is defined by {\unboldmath $a^2 = b\cos^2\delta$}, the blue line denotes the tangent line {\unboldmath $2\mu a = \cos\delta (b+\mu^2)$} and the red point marks the point of tangency. (b) An illustration of {\unboldmath $\mathcal{D}[\theta,\mu]$} (gray region), where the red point represents the center of {\unboldmath $\mathcal{D}[\theta,\mu]$}.}    
\label{fig:W_T_numerical_range_SRG_disk} 
\vspace{-0.5em} 
\end{figure}
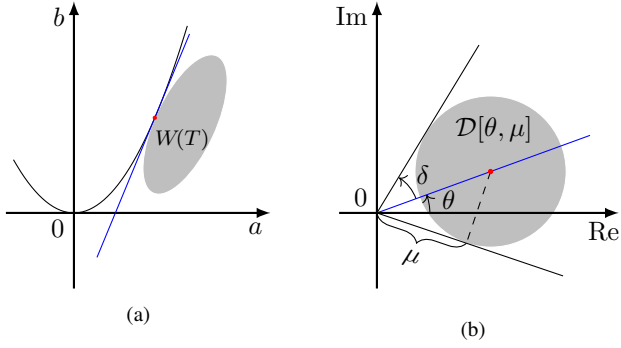 

Now we are ready to prove Theorem \ref{thm:segmental_phase_norm_minimization}. 
 
\proofof{Proof of Theorem \ref{thm:segmental_phase_norm_minimization}} 
    We first show (i). 
    Since $\Gamma_\theta(C) < \frac{\pi}{2}$, it follows from Proposition \ref{prop:theta_segmental_phase_LMI_graphical_interpretation} that $\Psi_\theta(C) \subset \theta + [-\delta, \delta]$ with $\delta \in [0,\frac{\pi}{2})$ if and only if there exists $\mu>0$ such that \eqref{eq:LMI_theta_segmental_phase} holds, which can be rewritten as 
    \begin{align*}
        &  \cos\delta(C^HC + \mu^2I) \leq \mu (e^{-j\theta} C + e^{j\theta}C^H) \\ 
       \Leftrightarrow  \qquad & \| e^{-j\theta}C - \frac{\mu}{\cos\delta}I\|\leq \mu \tan\delta \\ \Leftrightarrow \qquad & \| \gamma e^{-j\theta} C - I \| \leq \sin \delta , \quad \text{ with } \gamma = \frac{\cos\delta}{\mu}>0 . 
    \end{align*} Hence we have 
    \begin{equation*}
        \Gamma_\theta \leq \delta \quad \Leftrightarrow \quad \min_{\gamma>0} \|\gamma e^{-j\theta} C - I\| \leq \sin\delta , 
    \end{equation*}  implying 
    $\Gamma_\theta = \arcsin \min_{\gamma>0} \|\gamma e^{-j\theta} C - I\|$ and (i) holds. 
    
    To show (ii), by optimizing over $\theta$ to minimize the phase spread $\Gamma_\theta$, we have 
    \begin{align*}
        \Gamma^\star & \! = \! \arcsin \min_{\theta \in \mathbb{R}} \min_{\gamma>0} \|\gamma e^{-j\theta} C - I\| \! = \! \arcsin \min_{z \in \mathbb{C}} \|zC- I\|, \\  \theta^\star & \! = \! - \angle z^\star = - \angle   \arg \min_{z\in\mathbb{C}} \|zC-I\|. 
    \end{align*} 
    This completes the proof.  
\endproof

The connection between $\theta$-segmental phase and a norm minimization problem in Theorem \ref{thm:segmental_phase_norm_minimization} is of independent interest and brings several immediate insights. It enables the computation of $\theta$-segmental phase and the optimal segmental phase via SDP, thereby resolving a problem that was previously open~\cite{Chen2025SingularAngle}. Moreover, it provides the foundation for the interpolation construction developed later in the paper. 
A comparison of the SDP problems in Theorem \ref{thm:segmental_phase_norm_minimization} reveals a natural transition from the $0$-symmetric SRG to $\theta$-symmetric SRG and then to the canonical SRG defined by $\theta^\star$. 

\begin{remark} 
    The concept of segmental phase is closely related to the singular angle and antieigenvalue. Several computational aspects for the antieigenvalue have been investigated, one may refer to \cite{Davis1980Extending,Mirman1983Antieigenvalues,Paul2015Computation_theta_antieigenvalues,Gustafson1997NumericalRange} for more details. 
\end{remark} 

\begin{remark}
    There are some studies on the computation of the SRG that provide valuable insights~\cite{Pates2021Scaled,Huang2019SRG_Normal,Zhang2025DW_Shell}. Nevertheless, these methods are not directly applicable to the computation of the segmental phase. 
\end{remark}

When $\delta $ in Proposition \ref{prop:theta_segmental_phase_LMI_graphical_interpretation} exceeds $\frac{\pi}{2}$, we obtain the following equivalent characterizations of the $\theta$-segmental phase. Notably, this result bridges the phase condition with a semi-infinite parameterized family of matrix inequalities and further reduces it to a generalized eigenvalue problem (GEP).  

\begin{theorem}         \label{thm:larger_than_pi_spread}  
    Let matrix $C \in \mathbb{C}^{n\times n}$ and $\theta \in \mathbb{R}, \delta \in (\frac{\pi}{2},\pi]$ be given such that $0\notin \mathrm{SRG}_\theta(C)$, then the following statements are equivalent. 
    \begin{enumerate}[(i)]
        \item $\Psi_\theta(C) \subset \theta + (-\delta, \delta)$. 
        \item For any $\mu>0$, there holds 
        \begin{equation}        \label{eq:larger_than_pi_condition}
            \cos\delta(C^HC + \mu^2I) < \mu (e^{-j\theta} C + e^{j\theta}C^H) .  
        \end{equation} 
        \item The GEP $Mx \! = \! \mu Nx$ admits no eigenvalues in $ (0,\infty)$ where 
        \begin{equation}            \label{eq:GEP_matrices}
        M \! = \! \begin{bmatrix}
                - C^HC & 0 \\ 0  &  I  
            \end{bmatrix}\! , \ \  N \! = \! \begin{bmatrix}
                - \frac{2\mathrm{Re}\{ e^{-j\theta} C\}}{\cos\delta} &  I \\  I  & 0
            \end{bmatrix}\!  . 
        \end{equation}
    \end{enumerate}  
\end{theorem}

\begin{proof}
    See Appendix \ref{app:proofs}. 
\end{proof}

Combining Theorem \ref{thm:segmental_phase_norm_minimization} and Theorem \ref{thm:larger_than_pi_spread}, we obtain a complete characterization of the phase of $\theta$-symmetric SRG. Compared to the conventional approach of estimating the normalized numerical range for calculating the $\theta$-segmental phase~\cite{Chen2025SingularAngle}, the computational method here is significantly more efficient and straightforward to implement.

\section{$\theta$-Symmetric SRGs of Matrix Sums and Products}  
\label{section:subadd_submultiply} 

\subsection{Subadditive and Submultiplicative Properties}  

Having studied the $\theta$-symmetric SRG of a single matrix or a matrix set, we are interested in how the $\theta$-symmetric SRG behaves under the most fundamental matrix operations, namely sums and products. These two operations correspond, respectively, to parallel and cascade compositions, and they serve as the building blocks for more general interconnection structures.

We first introduce several concepts for preparation, which extend the chord and arc properties in the $0$-symmetric SRG framework. For given $\theta \in \mathbb{R}$, let $\mathrm{Chord}_\theta(z) $ denote the line segment connecting $z \in \mathbb{C}$ and its mirror reflection $\overline{z}e^{2j\theta}$ across the axis $e^{j\theta}\mathbb{R}$. Then a matrix set $\mathcal{C}$ is said to satisfy the $\theta$-chord property if  
\begin{equation*}
         \mathrm{Chord}_\theta(z) \subset \mathrm{SRG}_\theta(\mathcal{C}), \quad \forall z \in \mathrm{SRG}_\theta(\mathcal{C}). 
    \end{equation*} 
Similarly, let $\mathrm{Arc}_\theta^+ (z)$ and $\mathrm{Arc}_\theta^- (z)$ denote the circular arcs centered at the origin connecting $z$ and its reflection $\overline{z}e^{2j\theta}$ that cross the rays $e^{j\theta}\mathbb{R}^+$ and $e^{j\theta}\mathbb{R}^-$, respectively. Then $\mathcal{C}$ is said to satisfy the $\theta$-arc property if either of the following holds 
\begin{enumerate}[(i)]
    \item $ \mathrm{Arc_\theta^+}(z) \subset \mathrm{SRG}_\theta(\mathcal{C}), \quad \forall z \in \mathrm{SRG}_\theta(\mathcal{C}) $; 
    \item $ \mathrm{Arc_\theta^-}(z) \subset \mathrm{SRG}_\theta(\mathcal{C}), \quad \forall z \in \mathrm{SRG}_\theta(\mathcal{C}) $. 
\end{enumerate}

The following result establishes the subadditive and submultiplicative properties of the $\theta$-symmetric SRG, which describe how the $\theta$-symmetric SRG of the sum or product depends on its individual parts.  

\begin{theorem}             \label{thm:sum_product_properties}   
    Let matrix sets $\mathcal{A}_1,\mathcal{A}_2, \dots, \mathcal{A}_N \! \subset \! \mathbb{C}^{n\times n}$ be given, then 
    \begin{enumerate}[(i)]
        \item for each $\theta \in \mathbb{R}$ such that at least $N-1$ of the $\mathcal{A}_i$ satisfy the $\theta$-chord property, it follows  
        \begin{equation*}
            \mathrm{SRG}_\theta \left(  \sum_{i = 1}^N \mathcal{A}_i \right) \subset \sum_{i = 1}^N \mathrm{SRG}_{\theta} \left(\mathcal{A}_i\right) ;  
        \end{equation*}
        \item for each sequence $\theta_i \in \mathbb{R}, i = 1,2,\dots,N$, such that at least $N\!-\!1$ of the $\mathcal{A}_i$ satisfy the $\theta_i$-arc property, it follows 
        \begin{equation*}
            \mathrm{SRG}_{\sum_{i = 1}^N \theta_i} \left( \prod_{i = 1}^N \mathcal{A}_i \right) \subset  \prod_{i = 1}^N  \mathrm{SRG}_{\theta_i} \left( \mathcal{A}_i \right) . 
        \end{equation*} 
    \end{enumerate} 
\end{theorem}

\vspace{3pt} 

\begin{proof}
    We first show (i). Without loss of generality, suppose that $\mathcal{A}_k$ with $k \in \{ 1,2,\dots, N \}$ does not satisfy $\theta$-chord property, while $\mathcal{A}_i$ satisfies $\theta$-chord property for all $i \neq k$. 
    For all $z \in \mathrm{SRG}_\theta(\mathcal{A}_i)$ with $i\neq k$,  
\begin{equation*}
         \mathrm{Chord}_\theta(z) \subset \mathrm{SRG}_\theta(\mathcal{A}_i) = e^{j\theta}\mathrm{SRG}\left(e^{-j\theta}\mathcal{A}_i\right) , 
    \end{equation*} 
    implying 
    \begin{equation*}
        \mathrm{Chord}(ze^{-j\theta}) \subset \mathrm{SRG}\left(e^{-j\theta}\mathcal{A}_i\right). 
    \end{equation*}  
    Note that $z e^{-j\theta} \in \mathrm{SRG}\left(e^{-j\theta}\mathcal{A}_i\right)$, hence $e^{-j\theta}\mathcal{A}_i$ satisfies the chord property for all $i \neq k $. By Lemma \ref{lem:SRG_chord_arc_property}, we have 
    \begin{align*}
        \mathrm{SRG}_\theta \left( \sum_{i = 1}^N \mathcal{A}_i \right) & = e^{j\theta} \mathrm{SRG} \left(  \sum_{i = 1}^N e^{-j\theta} \mathcal{A}_i \right) \\ 
        & \subset e^{j\theta} \sum_{i = 1}^N \mathrm{SRG} \left( e^{-j\theta} \mathcal{A}_i \right)   \!=\! \sum_{i = 1}^N \mathrm{SRG}_\theta(\mathcal{A}_i) . 
    \end{align*} 
    Then we show (ii). Suppose, without loss of generality, that 
    $\mathcal{A}_k$ with $k \in \{ 1,2,\dots, N \}$  
    does not satisfy $\theta_k$-arc property, while $\mathcal{A}_i$ satisfies $\theta_i$-arc property for all $i \neq k$. For the case $i\neq k$, there are two possibilities: 
    $\mathrm{Arc_{\theta_i}^+}(z) \subset \mathrm{SRG}_{\theta_i}(\mathcal{A}_i)$ for all $z \! \in \! \mathrm{SRG}_{\theta_i}(\mathcal{A}_i)$ or $\mathrm{Arc_{\theta_i}^-}(z) \subset \mathrm{SRG}_{\theta_i}(\mathcal{A}_i)$ for all $z \! \in \!  \mathrm{SRG}_{\theta_i}(\mathcal{A}_i)$.  
    Both cases imply 
    $e^{-j\theta_i}\mathcal{A}_i$ satisfies the arc property. Denote $\theta = \sum_{i = 1}^N \theta_i$, then it follows from Lemma \ref{lem:SRG_chord_arc_property} that 
    \begin{align*}
        \mathrm{SRG}_{\theta}\! \left(  \prod_{i = 1}^N \mathcal{A}_i \right)     \! & \!= \! e^{j\theta} \mathrm{SRG} \left( \prod_{i = 1}^N e^{-j\theta_i} \mathcal{A}_i \right) \\ 
        & \! \subset \!  
        \prod_{i = 1}^N e^{j\theta_i} \mathrm{SRG}\left( e^{-j\theta_i}\mathcal{A}_i \right) \!=\! \prod_{i = 1}^N  \mathrm{SRG}_{\theta_i} ( \mathcal{A}_i ) . 
    \end{align*} 
    This completes the proof. 
\end{proof}

When considering the gain and phase of $\theta$-symmetric SRG separately, the gain aspect yields the classical result  
\begin{equation*}
    \overline{\sigma}\left( \sum_{i = 1}^N A_i \right) \leq \sum_{i = 1}^N \overline{\sigma}(A_i) , \quad \overline{\sigma} \left( \prod_{i = 1}^N A_i \right)  \leq \prod_{i = 1}^N \overline{\sigma}(A_i) , 
\end{equation*} 
while the phase aspect leads to the following corollary.  

\begin{corollary}
    Let $A_1,A_2,\dots,A_N\in\mathbb{C}^{n\times n}$ be given, then 
    \begin{enumerate}[(i)]  
        \item 
        for each $\theta \in \mathbb{R}$ such that $\Gamma_\theta(A_i) < \frac{\pi}{2}, i = 1,2,\dots,N$, it follows 
        \begin{equation*}
            \Psi_\theta \left( \sum_{i=1}^N A_i \right)  \subset \left[ \min_{1 \leq i \leq  N } \, \underline{\psi}_\theta(A_i) , \ \max_{1 \leq i \leq N }\, \overline{\psi}_\theta (A_i)  \right] ; 
        \end{equation*} 
        \item for each sequence $\theta_i \in \mathbb{R}, i = 1,2,\dots,N$, it follows 
        \begin{equation*}
            \Psi_{\sum_{i=1}^N \theta_i  } \! \left( \prod_{i = 1}^N A_i \right) \! \subset \! \sum_{i = 1}^N \Psi_{\theta_i}(A_i) .  
        \end{equation*}
    \end{enumerate}
\end{corollary}

\subsection{Nonsingularity of Return Difference Matrices} 

Consider the parallel and cascade block diagrams under unit negative feedback shown in Fig.~\ref{fig:parallel_cascade_diagram}. In the parallel case, the return ratio matrix is given by $\sum_{i=1}^N A_i$, whereas in the cascade case, it is given by $\prod_{i=1}^N A_i$.  
Accordingly, the return difference matrices are $I+\sum_{i=1}^N A_i$ and $I+\prod_{i=1}^N A_i$, respectively. We are interested in the nonsingularity of these sum-type and product-type return difference matrices, which is important in the stability analysis of feedback systems.

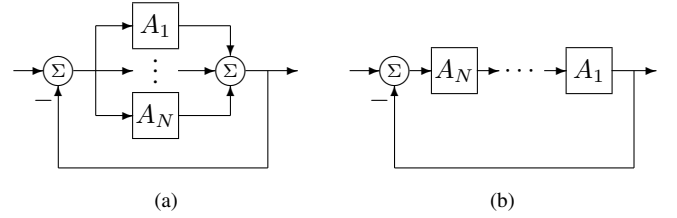
\begin{figure}[htb]  
\centering 
\subfigure[]{
\begin{minipage}[t]{0.47\linewidth} 
\setlength{\unitlength}{0.99mm} 
\centering 
\begin{picture}(42,21)(0,10)
    \put(0,24){\vector(1,0){4}} 
    \put(6,24){\circle{4}} 
    \put(6,24){\makebox(0,0){$\scriptstyle \Sigma$}} 
    \put(4,20){\makebox(0,0){$-$}} 
    \put(11,24){\line(0,1){6}} 
    \put(11,24){\line(0,-1){6}} 
    \put(8,24){\vector(1,0){8}} 
    \put(11,18){\vector(1,0){5}} 
    \put(11,30){\vector(1,0){5}} 
    \put(16,27){\framebox(6,6){$A_1$}} 
    \put(16,15){\framebox(6,6){$A_N$}} 
    \put(19,25){\makebox(0,0){$\vdots$}} 
    \put(22,30){\line(1,0){7}} 
    \put(22,18){\line(1,0){7}}
    \put(22,24){\vector(1,0){5}} 
    \put(29,30){\vector(0,-1){4}} 
    \put(29,18){\vector(0,1){4}} 
    \put(29,24){\circle{4}} 
    \put(29,24){\makebox(0,0){$\scriptstyle \Sigma$}} 
    \put(31,24){\vector(1,0){7}} 
    \put(34,24){\line(0,-1){13}} 
    \put(34,11){\line(-1,0){28}} 
    \put(6,11){\vector(0,1){11}}
\end{picture}   
\end{minipage} 
}\hfill%
\subfigure[]{
\begin{minipage}[t]{0.47\linewidth}
\setlength{\unitlength}{0.99mm} 
\centering 
\begin{picture}(42,21)(0,10)
    \put(0,24){\vector(1,0){4}}  
    \put(6,24){\circle{4}} 
    \put(6,24){\makebox(0,0){$\scriptstyle \Sigma$}} 
    \put(4,20){\makebox(0,0){$-$}} 
    \put(8,24){\vector(1,0){3}} 
    \put(11,21){\framebox(6,6){$A_N$}} 
    \put(17,24){\vector(1,0){3}} 
    \put(23,24){\makebox(0,0){$\cdots$}} 
    \put(26,24){\vector(1,0){3}}  
    \put(29,21){\framebox(6,6){$A_1$}} 
    \put(35,24){\vector(1,0){6}} 
    \put(38,24){\line(0,-1){13}} 
    \put(38,11){\line(-1,0){32}} 
    \put(6,11){\vector(0,1){11}}
\end{picture}   
\end{minipage} 
}  
\caption{(a) Parallel and (b) cascade interconnections under unit negative feedback. } 
\label{fig:parallel_cascade_diagram}  
\end{figure} 

Based on the subadditive and submultiplicative properties of $\theta$-symmetric SRG, we can derive graphical criteria for the nonsingularity of the sum-type and product-type return difference matrices.

\begin{theorem}          \label{thm:matrix_nonsingular_condition} 
    Let $A_1,A_2,\dots,A_N \in \mathbb{C}^{n\times n}$ be given, then 
    \begin{enumerate}[(i)] 
        \item $I+\sum_{i=1}^N A_i$ is nonsingular if there exists $\theta \in \mathbb{R}$ such that 
        \begin{equation}    \label{eq:matrix_nonsingular_sum_condition}
            -1 \notin \sum_{i = 1}^N \mathrm{SRG}_\theta(\mathcal{A}_i) , 
        \end{equation} 
        where each $\mathcal{A}_i$ is an arbitrary matrix set containing $A_i$, and at least $N-1$ of the $\mathcal{A}_i$ satisfy $\theta$-chord property; 
        \item $I+\prod_{i=1}^N A_i$ is nonsingular if there exists $\theta_i     \in \mathbb{R}, i = 1,2,\dots, N$, such that 
            \begin{equation}      \label{eq:matrix_nonsingular_product_condition} 
                 -1 \notin  \prod_{i=1}^N \mathrm{SRG}_{\theta_i}(\mathcal{A}_i),  
            \end{equation} 
    where each $\mathcal{A}_i$ is an arbitrary matrix set containing $A_i$, and at least $N-1$ of the $\mathcal{A}_i$ satisfy $\theta_i$-arc property. 
    \end{enumerate} 
\end{theorem}

\begin{proof}
    For (i), it follows from Theorem \ref{thm:sum_product_properties} that 
    \begin{align*}
        \Lambda\left(\sum_{i = 1}^N  A_i \right) & \subset \mathrm{SRG}_\theta \left( \sum_{i=1}^N A_i \right) \\ 
        & \subset \mathrm{SRG}_\theta \left( \sum_{i=1}^N \mathcal{A}_i \right) \subset \sum_{i =1}^N\mathrm{SRG}_\theta \left(  \mathcal{A}_i \right). 
    \end{align*} 
    In view of \eqref{eq:matrix_nonsingular_sum_condition}, we have $-1$ is not an eigenvalue of $\sum_{i=1}^NA_i$. Hence $I + \sum_{i=1}^NA_i$ is nonsingular. \\ 
    For (ii), by Theorem \ref{thm:sum_product_properties}, we have 
    \begin{align*}
        \Lambda\left(\prod_{i=1}^N A_i\right) &  
        \! \subset \! \mathrm{SRG}_{\sum_{i=1}^N \theta_i} \! \left( \prod_{i = 1}^N A_i \right) \\ 
        & \! \subset \! \mathrm{SRG}_{\sum_{i=1}^N  \theta_i} \!  \left( \prod_{i=1}^N \mathcal{A}_i \right) \!  
        \subset\!  \prod_{i = 1}^N \mathrm{SRG}_{\theta_i} ( \mathcal{A}_i ) . 
    \end{align*}  
     By \eqref{eq:matrix_nonsingular_product_condition}, we have $-1$ is not an eigenvalue of $\prod_{i=1}^N A_i$. Hence $I+\prod_{i=1}^N A_i$ is nonsingular, which completes the proof. 
\end{proof}

Fixing $\theta = 0, \theta_i = 0$ in Theorem \ref{thm:matrix_nonsingular_condition}, we obtain the following corollary, which coincides with the $0$-symmetric SRG result \cite{Chaffey2023Graphical}. The proof follows directly from Theorem \ref{thm:matrix_nonsingular_condition} and is omitted. 

\begin{corollary}                   \label{coro:nonsingularity_0_SRG}
    Let $A_1,A_2,\dots,A_N \in \mathbb{C}^{n\times n}$ be given, then 
    \begin{enumerate}[(i)]
        \item $I+\sum_{i=1}^NA_i$ is nonsingular if 
        $-1 \notin \sum_{i=1}^N \mathrm{SRG} (\mathcal{A}_i) , $ 
        where each $\mathcal{A}_i$ is an arbitrary matrix set containing $A_i$, and at least $N-1$ of the $\mathcal{A}_i$ satisfy the chord property; 
        \item $I+\prod_{i=1}^NA_i$ is nonsingular if 
        $-1 \notin  \prod_{i=1}^N \mathrm{SRG}(\mathcal{A}_i) , $ 
        where each $\mathcal{A}_i$ is an arbitrary matrix set containing $A_i$, and at least $N-1$ of the $\mathcal{A}_i$ satisfy the arc property. 
    \end{enumerate}
\end{corollary}

The following example demonstrates the superiority of Theorem \ref{thm:matrix_nonsingular_condition} over Corollary \ref{coro:nonsingularity_0_SRG}, revealing how the parameter flexibility $\theta$ provides a crucial degree of freedom to reduce conservatism.
\begin{example}
    We examine the sum-type and product-type cases separately as follows. 
    \begin{enumerate}[(i)]
        \item For the sum-type case, consider 
        $ A_1\!=\!A_2\!=\!A_3\!=\! \mathrm{diag} \{ 
            e^{j\frac{\pi}{3}},  e^{j\frac{2\pi}{3}} \} $.  
        One can readily verify that Corollary \ref{coro:nonsingularity_0_SRG} fails to determine the nonsingularity of $I+\sum_{i=1}^3 A_i $. In contrast, by choosing $\theta = \frac{\pi}{2}$, Theorem \ref{thm:matrix_nonsingular_condition} remains applicable and successfully guarantees its nonsingularity. 
        \item For the product-type case, consider 
        \begin{equation*}
            A_1 \! =\! \begin{bmatrix}
                e^{-j\frac{\pi}{6}} & 0 \\ 0 & e^{j\frac{\pi}{6}}
            \end{bmatrix}, \  A_2 \! = \! \begin{bmatrix}
                e^{j\frac{\pi}{3}} & 0 \\ 0 & e^{j\frac{2\pi}{3}}
            \end{bmatrix}, \  A_3 \! =\!  \overline{A}_2. 
        \end{equation*} Similarly, Corollary \ref{coro:nonsingularity_0_SRG} fails to determine the nonsingularity of $I+ \prod_{i=1}^3A_i$. However, Theorem \ref{thm:matrix_nonsingular_condition} remains effective to establish the nonsingularity by choosing $\theta_1 = 0, \theta_2 = \frac{\pi}{2}, \theta_3 = - \frac{\pi}{2}$. 
    \end{enumerate}  
\end{example}

\begin{remark} 
    Regarding Theorem~\ref{thm:matrix_nonsingular_condition} (ii), a critical open question is how to select the optimal $\vartheta = [\theta_1, \dots, \theta_N]^T \in \mathbb{R}^N$ that maximizes the distance between $-1$ and the product of the $\theta$-symmetric SRGs. This leads to the optimization problem:  
\begin{equation}                     \label{eq:find_optimal_theta}
    \vartheta^\star = \arg\max_{\vartheta \in \mathbb{R}^N} \min_{s_i \in \mathrm{SRG}_{\theta_i}(\mathcal{A}_i) } \left|1+\prod_{i=1}^N s_i \right| . 
\end{equation} 
This problem is analogous in spirit to the block-diagonal $D$-scaling used in structured singular value analysis \cite{Zhou1996Robust}, which minimizes $\overline{\sigma}(DMD^{-1})$ to tighten the upper bound of $\mu_\Delta(M)$. It is also closely related to the phase-based scaling problem investigated in \cite{Chen2026Cyclic}.  
Note that choosing $\theta_i$ to minimize the phase spread of each $\mathrm{SRG}_{\theta_i}(\mathcal{A}_i)$ separately does not, in general, solve the problem \eqref{eq:find_optimal_theta}. In fact, \eqref{eq:find_optimal_theta} 
is a mixed gain-phase optimization problem capturing the interactions among different $\mathrm{SRG}_{\theta_i}(\mathcal{A}_i)$, which is the subject of future work.  
\end{remark}

When gain and phase are considered separately, we obtain the following corollaries. The proofs are omitted for brevity. 

\begin{corollary}           \label{cor:separate_gain_phase_sum_prod_matrix}
     Let $A_1,A_2,\dots,A_N \in \mathbb{C}^{n\times n}$ be given, then 
     \begin{enumerate}[(i)]
         \item $I+\sum_{i=1}^N A_i$ is nonsingular if either $  \sum_{i = 1}^N \overline{\sigma}(A_i) < 1 $ holds, or there exists $\theta  \in \mathbb{R} $ such that $\Gamma_\theta(A_i) < \frac{\pi}{2}, i = 1,2,\dots,N$, and 
         \begin{equation*}
            \left[ \min_{1 \leq i \leq  N } \, \underline{\psi}_\theta(A_i) , \ \max_{1 \leq i \leq N }\, \overline{\psi}_\theta (A_i)  \right]  \subset (-\pi,\pi). 
        \end{equation*} 
        \item $I+\prod_{i=1}^N A_i$ is nonsingular if either $\prod_{i = 1}^N \overline{\sigma}(A_i) < 1 $ holds, or there exists $\theta_i \in \mathbb{R}, i = 1,2,\dots, N$, such that 
        \begin{equation*}
            \sum_{i=1}^N \Psi_{\theta_i}(A_i) \subset (-\pi,\pi). 
        \end{equation*} 
     \end{enumerate} 
\end{corollary}

\section{Robust Nonsingularity of $I+\prod_{i=1}^N A_i $} 
\label{section:robust_nonsingularity_matrices}

For Theorem \ref{thm:matrix_nonsingular_condition}, a natural question is whether the sufficient conditions are also necessary in the presence of uncertainties. In this part, we focus on the product-type and derive necessary and sufficient conditions for the nonsingularity of $I+ \prod_{i=1}^NA_i$ when some matrices are uncertain.  The sum-type $I+\sum_{i=1}^N A_i$ can be treated similarly and is relatively simpler since it only requires a common $\theta$ parameter. 

Since Theorem \ref{thm:matrix_nonsingular_condition} requires at least $N-1$ matrix sets to satisfy the $\theta$-chord or $\theta_i$-arc properties, we focus on a setting with $N-1$ uncertain matrices and one nominal matrix. Our goal is to derive necessary and sufficient conditions on the nominal matrix such that the nonsingularity conditions hold for all admissible choices of the $N-1$ uncertain matrices. Without loss of generality, assume the nominal matrix is $A_k$ with $k \in \{1,2,\dots,N\}$. For the index $k$, denote its complement as $\mathcal{K}' = \{ 1,2,\dots,N\} \!\setminus\! \{ k \}$.

We begin by characterizing the uncertainty sets. For given $0\leq\beta-\alpha\leq\pi$ and $\gamma>0$, denote $\theta = \frac{1}{2}(\beta+\alpha), \Gamma = \frac{1}{2}(\beta-\alpha)$ and define the following region 
\begin{align*}
    & \mathcal{R}[\alpha,\beta,\gamma] \! = \! \bigg\{  r e^{j(\theta+\delta)} \in\mathbb{C} : \\ 
    & \quad r \! \leq \!  \frac{\gamma}{1+\sin\Gamma} \left(  \cos\delta \! +\! \sqrt{\sin^2\Gamma - \sin^2\delta} \right) , \delta \! \in \! [-\Gamma,\Gamma]   \bigg\}  .  
\end{align*}
An illustration of $\mathcal{R}[\alpha,\beta,\gamma]$ is shown in Fig.~\ref{subfig:R_region}. When $\alpha = -\beta$, we write $\mathcal{R}[-\beta,\beta,\gamma]$ simply as $\mathcal{R}[\beta,\gamma]$.

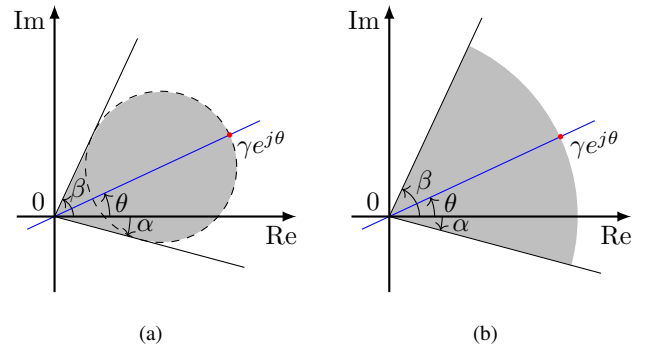
\begin{figure}[htb] 
\centering 
\subfigure[]{
\begin{minipage}[t]{0.46\linewidth}
\centering 
\begin{tikzpicture}
\begin{scope}  
\draw [white, fill = gray, opacity = 0.5] (0.8,1.2) -- ++(-15:1.2) arc (-105:155:1) -- cycle; 
\draw [dashed] ($(0.8,1.2)+(25:1.555)$) circle(1); 
\draw (0.8,1.2) -- ++(65:2.6);
\draw (0.8,1.2) -- ++(-15:2.6); 
\draw [blue] (0.8,1.2) -- ++(25:3.0); 
\draw [blue] (0.8,1.2) -- ++(-155:0.4); 
\draw[->] (1.05,1.2) arc(0:65:0.25); 
\draw[->] (1.53,1.2) arc(0:25:0.73); 
\draw[->] (1.8,1.2) arc(0:-15:1);  
\end{scope} 
\draw[-latex, thick](0.3,1.2)--(4,1.2); 
\draw[-latex, thick](0.8,0.2)--(0.8,4); 
\draw [red, fill=red] ($(0.8,1.2)+(25:2.555)$) circle(0.8pt); 
\node[right] at (3.09,2.15){$\gamma e^{j\theta}$}; 
\node[below] at (3.8,1.2){$\mathrm{Re}$}; 
\node[left] at(0.8,3.8){$\mathrm{Im}$}; 
\node[left] at(0.8,1.4){$0$}; 
\node[right] at (0.88,1.55){$\beta$}; 
\node[right] at (1.48,1.38){$\theta$}; 
\node[below] at (2,1.28){$\alpha$};   
\end{tikzpicture} 
\label{subfig:R_region}
\end{minipage} 
}
\subfigure[]{
\begin{minipage}[t]{0.46\linewidth}
\centering 
\begin{tikzpicture}
\begin{scope}  
\draw [white, fill = gray, opacity = 0.5] (0.8,1.2) -- ++(-15:2.5) arc (-15:65:2.5) -- cycle; 
\draw (0.8,1.2) -- ++(65:2.9);
\draw (0.8,1.2) -- ++(-15:2.9); 
\draw [blue] (0.8,1.2) -- ++(25:3.0); 
\draw [blue] (0.8,1.2) -- ++(-155:0.4); 
\draw[->] (1.2,1.2) arc(0:65:0.4); 
\draw[->] (1.4,1.2) arc(0:25:0.6); 
\draw[->] (1.5,1.2) arc(0:-15:0.7);  
\end{scope} 
\draw[-latex, thick](0.3,1.2)--(4,1.2); 
\draw[-latex, thick](0.8,0.2)--(0.8,4); 
\draw [red, fill=red] ($(0.8,1.2)+(25:2.5)$) circle(0.8pt); 
\node[right] at (3.07,2.15){$\gamma e^{j\theta}$}; 
\node[below] at (3.8,1.2){$\mathrm{Re}$}; 
\node[left] at(0.8,3.8){$\mathrm{Im}$}; 
\node[left] at(0.8,1.4){$0$}; 
\node[right] at (1.03,1.65){$\beta$}; 
\node[right] at (1.4,1.38){$\theta$}; 
\node[below] at (1.75,1.28){$\alpha$};    
\end{tikzpicture} 
\label{subfig:S_region} 
\end{minipage} 
} 
\vspace{-0.5em}  
\caption{Illustrations of (a) {\unboldmath$\mathcal{R}[\alpha,\beta,\gamma]$} (gray region), and (b) {\unboldmath$\mathcal{S}[\alpha,\beta,\gamma]$} (gray region), where {\unboldmath $\theta = \frac{1}{2}(\beta+\alpha)$} and the red point represents the complex scalar {\unboldmath $\gamma e^{j\theta}$}.}    
\label{fig:R_and_S_region} 
\vspace{-0.5em} 
\end{figure} 

Accordingly, graphically define the matrix set via its canonical SRG as 
\begin{equation*}               
    \mathcal{M}[\alpha,\beta,\gamma] = \{ C\in\mathbb{C}^{n\times n}: \mathrm{SRG}_{\theta^\star}(C) \subset \mathcal{R}[\alpha,\beta,\gamma] \} . 
\end{equation*} 
Note that for different elements in $\mathcal{M}[\alpha,\beta,\gamma]$, the corresponding phase centers $\theta^\star$ of the canonical SRG are generally different. Alternatively, by prescribing a uniform phase center, we can define another matrix set as 
\begin{align*}
    \mathcal{M}_c[\alpha,\beta,\gamma] \! = \! \{  C \! \in \!  \mathbb{C}^{n\times n} \!: \! \mathrm{SRG}_{\frac{1}{2}(\alpha+\beta)}(C) \! \subset \! \mathcal{R}[\alpha,\beta,\gamma] \} , 
\end{align*} 
The following result gives a unified necessary and sufficient condition for the nonsingularity of $I+\prod_{i=1}^NA_i$ under these two uncertainty descriptions.  

\begin{theorem}                     \label{thm:matrix_iff_condition} 
    Let $0\leq \beta_i - \alpha_i\leq \pi, \gamma_i >0, i \in  \mathcal{K}'$, be given, then $I+\prod_{i=1}^NA_i$ is nonsingular for all $A_i \in \mathcal{M}[\alpha_i,\beta_i,\gamma_i], i \in \mathcal{K}'$, or respectively, for all $  A_i \in \mathcal{M}_c[\alpha_i,\beta_i,\gamma_i], i \in \mathcal{K}'$, if and only if there exists $\theta\in\mathbb{R}$ such that 
    \begin{equation}                    \label{eq:matrix_nonsingularity_iff_condition}
        -1 \notin  \mathrm{SRG}_{\theta} (A_k) \prod_{i \in \mathcal{K}'} \mathcal{R} [\alpha_i,\beta_i,\gamma_i] . 
    \end{equation} 
\end{theorem}

\begin{proof}
    See Appendix \ref{app:proofs}. 
\end{proof}

When only independent gain and phase bounds are available for the uncertainty, rather than the refined gain-phase coupling captured by $\mathcal{R}[\alpha,\beta,\gamma]$, the following region on the complex plane serves as a useful characterization. 
\begin{equation*}
    \mathcal{S} [\alpha,\beta,\gamma] = \{ z \in \mathbb{C}: \angle z \in [\alpha,\beta], |z| \leq \gamma \}. 
\end{equation*} 
An illustration of $\mathcal{S}[\alpha,\beta,\gamma]$, which is a sector-shaped region, is shown in Fig.~\ref{subfig:S_region}. Compared with $\mathcal{R}[\alpha,\beta,\gamma]$, there holds 
$ \mathcal{S}[\alpha,\beta,\gamma] \supset \mathcal{R}[\alpha,\beta,\gamma]  $.  
If $\alpha = - \beta$, we write $\mathcal{S}[-\beta,\beta,\gamma]$ simply as $\mathcal{S}[\beta,\gamma]$. We can graphically define the following matrix set via the canonical SRG as  
\begin{equation*}
    \mathcal{N}  [\alpha,\beta,\gamma] \! = \! \{ C \! \in \!  \mathbb{C}^{n\times n} \! : \! \mathrm{SRG}_{\theta^\star} (C) \! \subset \!  \mathcal{S}[\alpha,\beta,\gamma] \} . 
\end{equation*}

For the uncertainty set $\mathcal{N} [\alpha,\beta,\gamma]$, a necessary and sufficient condition for the nonsingularity of $I+\prod_{i=1}^NA_i$ can be established in a similar manner. The proof is analogous to that of Theorem \ref{thm:matrix_iff_condition} and is omitted for brevity. 

\begin{proposition}
    Let $0\leq \beta_i - \alpha_i \leq \pi, \gamma_i >0, i \in \mathcal{K}'$, be given, then $I+ \prod_{i=1}^NA_i$ is nonsingular for all $A_i \in \mathcal{N} [\alpha_i,\beta_i,\gamma_i], i  \! \in \! \mathcal{K}'$, if and only if there exists $\theta \! \in \!  \mathbb{R}$ such that 
    \begin{equation*}
        -1 \notin \mathrm{SRG}_\theta(A_k) \mathcal{S} \left[ \sum_{i \in \mathcal{K}'} \alpha_i, \sum_{i \in \mathcal{K}'} \beta_i, \prod_{i \in \mathcal{K}'} \gamma_i \right] . 
    \end{equation*}
\end{proposition}

\vspace{5pt} 

When $\gamma \to \infty$, $\mathcal{M}[\alpha,\beta,\gamma]$ and $\mathcal{N} [\alpha,\beta,\gamma]$ coincide and reduce to a matrix cone defined by 
\begin{equation*}
    \mathcal{C} [\alpha,\beta] = \{ C \in \mathbb{C}^{n \times n}: \Psi_{\theta^\star} (C) \subset [\alpha,\beta] \} . 
\end{equation*}
In this case, we obtain the following corollary, which concerns only the phase-bounded uncertainties. The proof follows directly from Theorem \ref{thm:matrix_iff_condition} and is omitted. 

\begin{corollary}
    Let $0 \leq \beta_i-\alpha_i\leq \pi,  i \in \mathcal{K}'$, be given, then $I+\prod_{i=1}^NA_i$ is nonsingular for all $A_i\in \mathcal{C} [\alpha_i,\beta_i] , i \in \mathcal{K}'$, if and only if there exists $\theta \in \mathbb{R}$ such that 
    \begin{equation*}
        \Psi_\theta(A_k) + \sum_{i \in \mathcal{K}'} [\alpha_i,\beta_i] \subset (-\pi,\pi) . 
    \end{equation*}
\end{corollary}

\section{Robust Stability Analysis of Cyclic Cascaded Systems} 
\label{section:single_loop_stability_systems}

For a system $G \in \mathcal{RH}_\infty^{m\times m}$, by plotting the $\theta$-symmetric SRG of $G(j\omega)$ at each frequency, we obtain an over-approximation of the eigenloci. 
A typical choice of $\theta(\omega)$ is to apply Theorem \ref{thm:segmental_phase_norm_minimization} frequency-wisely to get $\theta^\star(\omega)$. By selecting a continuous $\theta^\star(\omega)$, this yields the frequency-wise canonical SRG, which serves as a natural MIMO extension of the Nyquist plot. 
Compared with the MIMO eigenloci, the advantage of the frequency-wise $\theta$-symmetric SRG lies in its compositionality: the properties of loop transfer matrix can be directly deduced from those of the subsystems, whereas its eigenloci do not admit such a direct compositional characterization.  
An illustrative example of the frequency-wise canonical SRG of a system is provided below.  

\begin{example}                     \label{examp:system_theta_SRG_example}
    Consider a system 
    \begin{equation}                \label{eq:system_example} 
        G(s) = \begin{bmatrix}
            \frac{s+2}{s+1} & 0 \\ 1 & \frac{s^2+3s+27}{2s^2+4s+6} 
        \end{bmatrix} . 
    \end{equation}  
    By solving the SDP formulated in Theorem \ref{thm:segmental_phase_norm_minimization} at each frequency, we compute the optimal $\theta^\star(\omega)$. 
    Then the frequency-wise canonical SRG and eigenloci of $G(s)$ are shown in Fig.~\ref{fig:system_frequency_wise_theta_SRG}. From the plot, one can see that $G(s)$ exhibits an overall phase-lag behavior, which the $0$-symmetric SRG method fails to capture. 
\end{example}

\begin{figure}[htb]
    \centering
    \begin{overpic}[width=0.5\linewidth]{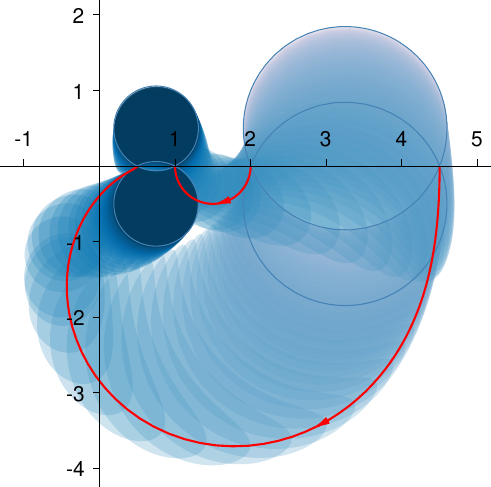}
        \put(120,75){\footnotesize$\mathrm{Re}$} 
        \put(28,120){\footnotesize$\mathrm{Im}$}  
    \end{overpic}
    \caption{An illustration of the frequency-wise  
    canonical SRG of {\unboldmath$G(s)$} in \eqref{eq:system_example}, where the red curves denote the eigenloci of {\unboldmath$G(s)$}. }
    \label{fig:system_frequency_wise_theta_SRG} 
\end{figure} 

For the SISO case, the frequency-wise canonical SRG degenerate precisely to the classical Nyquist plot. The proof follows directly from the definition of canonical SRG. 

\begin{proposition}             \label{prop:reduction_to_SISO_Nyquist_plot} 
    For a SISO system $g \in \mathcal{RH}_\infty$, there holds 
    \begin{equation*}
        \mathrm{SRG}_{\theta^\star(\omega)} (g(j\omega)) = g(j\omega), \quad \forall \omega \in [0,\infty], 
    \end{equation*} 
    where $\theta^\star(\omega) = \angle g(j\omega)$ is continuous. 
\end{proposition}

\subsection{Cyclic Stability Criteria via $\theta$-Symmetric SRG}

Consider the feedback interconnection of a cascade of systems shown in Fig.~\ref{fig:cyclic_feedback_system}.  
When $P_i\in \mathcal{RH}_\infty^{m\times m}$, the cyclic feedback system is stable if and only if \cite{Zhou1996Robust,Chen2026Cyclic} 
\begin{equation}                \label{eq:to_show_cyclic_stability_condition}
    (I+P_1P_2\cdots P_N)^{-1} \in \mathcal{RH}_\infty^{m \times m}. 
\end{equation} 

Now we are ready to show the main result in this section, which provides a graphical stability criterion for the feedback interconnection of cascaded systems via $\theta$-symmetric SRG.

\begin{theorem}                         \label{thm:cyclic_system_theta_SRG}
    The feedback interconnection of a cascade of systems $P_1,P_2,\dots,P_N\in \mathcal{RH}_\infty^{m \times m}$ in Fig.~\ref{fig:cyclic_feedback_system} is stable if there exist functions $\theta_i(\omega) \in \mathbb{R}, i = 1, 2, \dots, N$, such that for all $\omega \in [0,\infty]$,   
    \begin{equation}            \label{eq:cyclic_system_stability_condition}
        [-\infty, -1]  \cap   \prod_{i = 1}^N \mathrm{SRG}_{\theta_i(\omega)}(\mathcal{P}_i(j\omega)) = \emptyset,  
    \end{equation}  
    where each $\mathcal{P}_i(j\omega)$ is an arbitrary matrix set containing $P_i(j\omega)$, and at least $N-1$ of the $\mathcal{P}_i(j\omega)$ satisfy $\theta_i(\omega)$-arc property. 
\end{theorem}

\begin{proof}
    Since $P_i \in \mathcal{RH}_\infty^{m \times m}$, it suffices to show condition \eqref{eq:to_show_cyclic_stability_condition} holds. It follows from the generalized Nyquist criterion \cite{Desoer1980Generalized,Griggs2012InterconnectionsMixed} that \eqref{eq:to_show_cyclic_stability_condition} holds if 
    \begin{equation}           \label{eq:sufficient_nyquist_condition}
        I + \tau P_1(j\omega)P_2(j\omega)\cdots P_N(j\omega) \ \text{ is nonsingular,   } 
    \end{equation} 
    for all $\omega \in [0,\infty]$ and all $\tau\in[0,1]$. \\ 
    By Theorem \ref{thm:matrix_nonsingular_condition}, \eqref{eq:cyclic_system_stability_condition} implies that for all $\omega \in [0,\infty]$, 
    \begin{equation*}
        [-\infty, -1] \cap \Lambda(P_1(j\omega)P_2(j\omega)\cdots P_N(j\omega)) =\emptyset. 
    \end{equation*}  
    This indicates that \eqref{eq:sufficient_nyquist_condition} holds. 
    Hence the cyclic feedback system is stable and the proof is completed.  
\end{proof} 

\begin{figure}[htb] 
    \setlength{\unitlength}{1mm}
		\begin{center}
			\begin{picture}(84,20)
				\thicklines 
                \put(0,14){\vector(1,0){6}}
                \put(7,14){\circle{2}} 
                \put(8,14){\vector(1,0){6}} 
                \put(14,10){\framebox(8,8){$P_N$}} 
                \put(22,14){\vector(1,0){5}} 
                \put(30,14){\makebox(0,0){$\cdots$}} 
                \put(33,14){\vector(1,0){6}} 
                \put(40,14){\circle{2}} 
                \put(40,20){\vector(0,-1){5}} 
                \put(41,14){\vector(1,0){6}} 
                \put(47,10){\framebox(8,8){$P_2$}} 
                \put(55,14){\vector(1,0){5}} 
                \put(61,14){\circle{2}} 
                \put(61,20){\vector(0,-1){5}} 
                \put(62,14){\vector(1,0){6}} 
                \put(68,10){\framebox(8,8){$P_1$}} 
                \put(76,14){\vector(1,0){7}} 
                \put(79,14){\line(0,-1){11}} 
                \put(79,3){\line(-1,0){72}} 
                \put(7,3){\vector(0,1){10}} 
                \put(3,9){\makebox(3,3){$-$}} 
                \put(1.5,15){\makebox(3,3){$e_N$}} 
                \put(9,10){\makebox(3,3){$u_N$}} 
                \put(36,18){\makebox(3,3){$e_2$}} 
                \put(42.5,10){\makebox(3,3){$u_2$}} 
                \put(57,18){\makebox(3,3){$e_1$}} 
                \put(63.5,10){\makebox(3,3){$u_1$}} 
			\end{picture}
			\vspace{-0.5em} 
                \caption{Feedback interconnection of a cascade of systems. }
			\label{fig:cyclic_feedback_system} 
		\end{center} 
        \vspace{-0.5em} 
\end{figure}
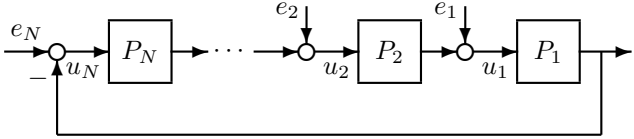

By fixing $\theta_i(\omega) = 0$, Theorem \ref{thm:cyclic_system_theta_SRG} reduces to the following $0$-symmetric SRG result. The proof is omitted for brevity. 

\begin{corollary}           \label{cor:standard_SRG_stability_system}
    The feedback interconnection of a cascade of systems $P_1,P_2,\dots,P_N \in \mathcal{RH}_\infty^{m\times m}$ in Fig.~\ref{fig:cyclic_feedback_system} is stable if for all $\omega \in [0,\infty]$, 
    \begin{equation*}
        [-\infty, -1] \cap  \prod_{i=1}^N\mathrm{SRG}(\mathcal{P}_i(j\omega)) = \emptyset, 
    \end{equation*}  
    where each $\mathcal{P}_i(j\omega)$ is an arbitrary matrix set containing $P_i(j\omega)$, and at least $N-1$ of the $\mathcal{P}_i(j\omega)$ satisfy arc property. 
\end{corollary}

Here we present a simple illustrative example to clarify the gap between Theorem \ref{thm:cyclic_system_theta_SRG} and Corollary \ref{cor:standard_SRG_stability_system}, which highlights the benefits of the $\theta$-symmetric SRG compared with the $0$-symmetric SRG.  

\begin{example}
    Consider two SISO transfer functions 
    \begin{equation*}
        p_1(s) = \frac{s^2+4s+36}{s^2+3s+5}, \quad p_2(s) = \frac{1}{p_1(s) } = \frac{s^2+3s+5}{s^2+4s+36} . 
    \end{equation*}  
    Clearly, $\left(1+p_1(s)p_2(s)\right)^{-1} \in \mathcal{RH}_\infty$, i.e, the feedback system of $p_1(s)$ and $p_2(s)$ is stable. We next show that Corollary \ref{cor:standard_SRG_stability_system} is not applicable to certify the feedback stability, whereas Theorem \ref{thm:cyclic_system_theta_SRG} successfully certifies it. 
    
    Note that for $\omega_0 = 3$, we have  
    \begin{equation*}
        \mathrm{SRG}(p_1(j\omega_0)) \! = \!  \{  -3j, 3j \}, \ \  \mathrm{SRG}(p_2(j\omega_0)) \! = \!  \left\{ \frac{1}{3}j,-\frac{1}{3}j  \right\} , 
    \end{equation*}
    implying 
    \begin{equation*} 
        [-\infty, -1] \cap  \mathrm{SRG}(p_1(j\omega_0)) \mathrm{SRG}(p_2(j\omega_0))   =  \{ -1 \} \neq \emptyset. 
    \end{equation*}  
    Hence, Corollary \ref{cor:standard_SRG_stability_system} is not applicable.  
    Now we show the effectiveness of Theorem \ref{thm:cyclic_system_theta_SRG}. 
    Let $\theta_1(\omega) = \angle p_1(j\omega), \theta_2(\omega) = \angle p_2(j\omega)$, then by Proposition \ref{prop:reduction_to_SISO_Nyquist_plot},  
    \begin{equation*}
         [-\infty,-1]  \cap \mathrm{SRG}_{\theta_1(\omega)}(p_1(j\omega)) \mathrm{SRG}_{\theta_2(\omega)}(p_2(j\omega))  = \emptyset  . 
    \end{equation*}   
    Since $p_2(j\omega)$ itself satisfies the $\theta_2(\omega)$-arc property, it follows from Theorem \ref{thm:cyclic_system_theta_SRG} that the feedback system is stable. 
\end{example}

When the gain and phase are considered separately, we obtain the following result, which can be reduced to the results in \cite{Chen2026Cyclic}. It follows from Theorem \ref{thm:cyclic_system_theta_SRG} and the proof is omitted. 

\begin{proposition}               \label{prop:cyclic_system_mixed_gain_phase}
        The feedback interconnection of a cascade of systems $P_1,P_2,\dots,P_N \in \mathcal{RH}_\infty^{m\times m}$ in Fig.~\ref{fig:cyclic_feedback_system} is stable if for each $\omega \in [0,\infty]$, one of the following conditions holds: 
        \begin{enumerate}[(i)]
            \item $\displaystyle  \prod_{i=1}^N \overline{\sigma}(P_i(j\omega)) < 1$; 
            \item there exist $\theta_i(\omega) \in \mathbb{R}, i = 1,2,\dots,N$, such that 
            \begin{equation*}
                \sum_{i=1}^N \Psi_{\theta_i(\omega)} (P_i(j\omega)) \subset (-\pi,\pi) . 
            \end{equation*}
        \end{enumerate}
\end{proposition}

\subsection{Necessary and Sufficient Conditions for Robust Stability} 
\label{subsection:necessary_sufficient_conditions_systems}

For the cyclic interconnection shown in Fig.~\ref{fig:cyclic_feedback_system}, consider the case where some subsystems are uncertain and belong to a prescribed uncertainty set. Our objective is to establish the necessary and sufficient conditions on the nominal system to guarantee robust closed-loop stability.
In particular, we focus on the following set of uncertain systems, whose $0$-symmetric SRGs satisfy mixed gain and phase constraints. 
\begin{align*}
    & \mathfrak{P}[\beta,\gamma] \!=\! \{ P \! \in \! \mathcal{RH}_\infty^{m\times m} \! : \\  & \qquad \qquad \qquad  \quad  \mathrm{SRG}(P(j\omega)) \! \subset \! \mathcal{R}[\beta,\gamma],  \ \forall \omega \! \in \! [0,\infty] \} . 
\end{align*} 

To apply the matrix-level necessary and sufficient conditions in Theorem \ref{thm:matrix_iff_condition} to dynamical systems, the critical step lies in solving an interpolation problem associated with the $\theta$-symmetric SRG. 
Analogous to classical interpolation problems in the small-gain and passivity frameworks \cite{Ball1990Interpolation}, this is a fundamental yet considerably more challenging task, as the interpolating transfer matrix must satisfy simultaneous gain and phase constraints. Although the general problem remains largely open, we develop a constructive method for a specific class of cases, enabling us to derive exact robust stability conditions for cyclic interconnections. Specifically, the following result demonstrates that the single-frequency interpolation problem is solvable under the mixed constraints induced by the $0$-symmetric SRG. The proof is established via the explicit construction of the corresponding transfer matrix.

\begin{proposition}             \label{prop:interpolation_result}
    Let $\omega_0\in[0,\infty], \beta \in [0,\frac{\pi}{2}], \gamma > 0$ and $\Delta \in \mathbb{C}^{m\times m}$ be given, if 
    \begin{equation}                    \label{eq:interpolation_condition}
        \mathrm{SRG}(\Delta) \subset  \mathcal{R}[\beta,\gamma], 
    \end{equation} 
    then there exists a transfer matrix $P \in \mathfrak{P}[\beta,\gamma]$ such that $P(j\omega_0) = \Delta$. 
\end{proposition}

\begin{proof}
    See Appendix \ref{app:proofs}. 
\end{proof}

\begin{remark}
    A natural extension involves uncertainties characterized by $\theta$-symmetric SRGs with asymmetric phase bounds. However, solving the corresponding interpolation problem is nontrivial and does not follow directly from Proposition \ref{prop:interpolation_result}. The main challenge lies in ensuring that the constructed interpolating transfer matrix has real coefficients. This direction is left for future research. 
\end{remark}

As $\gamma \to \infty$, $\mathfrak{P}[\beta,\gamma]$ degenerates to the cone $\mathfrak{C}[\beta]$ given by 
\begin{equation*} 
    \mathfrak{C}[\beta] \! = \! \{ P \! \in \! \mathcal{RH}_\infty^{m \times m} \! :  
    \Psi_0(P(j\omega))  \! \subset \! [-\beta,\beta],  \ \forall \omega  \! \in \!  [0,\infty]  \} . 
\end{equation*}
Consequently, Proposition \ref{prop:interpolation_result} yields the following result, which provides a solution to the interpolation problem under phase constraints alone. The proof follows directly from Proposition \ref{prop:interpolation_result} and is omitted. 

\begin{corollary}
    Let $\omega_0\in[0,\infty], \beta \in [0,\frac{\pi}{2}]$ and $\Delta \in \mathbb{C}^{m\times m}$ be given, if $\Psi_0(\Delta)\subset[-\beta,\beta]$, then there exists a transfer matrix $P \in \mathfrak{C}[\beta]$ such that $P(j\omega_0) = \Delta$.  
\end{corollary}

We are now ready to state the necessary and sufficient conditions for the robust stability of cyclic interconnection in Fig.~\ref{fig:cyclic_feedback_system} when some subsystems are uncertainties. Without loss of generality, assume $P_k(s)$ is the nominal system, while $P_i(s), i \in \mathcal{K}'$, represent uncertain systems.

\begin{theorem}                 \label{thm:system_iff_condition} 
    Let $\beta_i \in [0,\frac{\pi}{2}], \gamma_i > 0, i \in \mathcal{K}'$, be given, then the feedback interconnection of a cascade of systems $P_1,P_2,\dots,P_N \in \mathcal{RH}_\infty^{m\times m}$ in Fig.~\ref{fig:cyclic_feedback_system} is stable for all $P_i \in \mathfrak{P}[\beta_i,\gamma_i], i  \in \mathcal{K}' $, if and only if there exists a function $\theta(\omega) \in \mathbb{R}$ such that 
    \begin{equation*}
        -1 \notin \mathrm{SRG}_{\theta(\omega)} (P_k(j\omega)) \prod_{i \in \mathcal{K}'} \mathcal{R}[\beta_i,\gamma_i] , \quad \forall \omega \in [0,\infty]. 
    \end{equation*}
\end{theorem}

\vspace{0.5em} 

\begin{proof}
    The sufficiency follows from Theorem \ref{thm:cyclic_system_theta_SRG} and we show the necessity here. Suppose by contradiction that for any $\theta(\omega)$, there always exists $\omega_0 \in [0,\infty]$ such that 
    \begin{equation*}
        -1 \in \mathrm{SRG}_{\theta(\omega_0)} (P_k(j\omega_0)) \prod_{i \in \mathcal{K}'} \mathcal{R}[\beta_i,\gamma_i] . 
    \end{equation*} 
    Then it follows from Theorem~\ref{thm:matrix_iff_condition} that there exist matrices $\Delta_i $ such that 
    \begin{equation*}
       \mathrm{SRG}(\Delta_i) \subset \mathcal{R}[\beta_i,\gamma_i], \quad i \in \mathcal{K}' 
    \end{equation*} 
    and $I+ \Delta_1  \cdots \Delta_{k-1} P_k(j\omega_0) \Delta_{k+1} \cdots  \Delta_N$ is singular. Then by Proposition \ref{prop:interpolation_result}, there exist $P_i\in\mathcal{RH}_\infty^{m\times m}$ such that 
    \begin{equation*}
        P_i \in \mathfrak{P}[\beta_i,\gamma_i],  \ \  \text{ and } \ \  P_i(j\omega_0) = \Delta_i, \quad i \in \mathcal{K}'. 
    \end{equation*}
    This indicates that for these $P_i(s)$, we have 
    \begin{equation*}
        I + P_1(j\omega_0) \cdots P_k(j\omega_0) \cdots P_N(j\omega_0) \ \text{ is singular, } 
    \end{equation*} 
    implying the feedback system is not stable, which is a contradiction. This completes the proof. 
\end{proof}

Focusing solely on phase uncertainties, the following result is a direct consequence of Theorem \ref{thm:system_iff_condition}. It establishes a necessary and sufficient condition for robust stability under phase-bounded uncertainties.

\begin{corollary}
    Let $\beta_i \in [0,\frac{\pi}{2}], i \in \mathcal{K}'$, be given, then the feedback interconnection of a cascade of systems $P_1,P_2,\dots,P_N \in \mathcal{RH}_\infty^{m\times m}$ in Fig.~\ref{fig:cyclic_feedback_system} is stable for all $P_i \in \mathfrak{C}[\beta_i], i \in \mathcal{K}'$, if and only if there exists a function $\theta(\omega)\in\mathbb{R}$ such that 
    \begin{equation*}
        \Psi_{\theta(\omega)}(P_k(j\omega)) \! + \! \sum_{i \in \mathcal{K}'} [-\beta_i,\beta_i] \! \subset \! (-\pi,\pi) , \quad \forall \omega \! \in \!  [0,\infty] . 
    \end{equation*}
\end{corollary}

\begin{remark}
    While derived for stable subsystems, analogous results hold for semi-stable cases provided the $\theta$-symmetric SRG is evaluated on a generalized Nyquist contour. 
\end{remark}

\section{Cactus Network Stability Analysis} 
\label{section:cactus_graph}  

In Section \ref{section:single_loop_stability_systems}, we have performed stability analysis for the cyclic interconnection, corresponding to the single loop case from a network perspective. We now generalize the $\theta$-symmetric SRG framework to multi-loop networks with cactus structure. Such cactus networks arise in a wide range of applications \cite{Arcak2011DiagonalStabilityCactus,Azuma2019CactusExpandable,Duan2021GraphStability}. 

Let $\mathbb{G} = (\mathcal{V},\mathcal{E})$ be a directed graph representation of an interconnected network with vertices $\mathcal{V} $ and edges $\mathcal{E} \subset \mathcal{V}\times \mathcal{V}$. A path is a sequence of edges that connect a sequence of vertices.  
A path is called a loop if the starting and ending vertices are the same. A loop is called a simple loop if there are no repeated vertices other than the starting and ending vertices. The graph $\mathbb{G}$ is said to be strongly connected if for any $i,j\in\mathcal{V}$, there exists a path from vertex $i$ to vertex $j$. We focus on graphs with a cactus structure, defined as follows. 

\begin{definition}[\!\! \cite{Arcak2011DiagonalStabilityCactus}]
    A graph $\mathbb{G}=(\mathcal{V},\mathcal{E})$ is said to be a cactus graph, if $\mathbb{G}$ is strongly connected and each pair of distinct simple loops have at most one common vertex. 
\end{definition} 

For a cactus graph $\mathbb{G}=(\mathcal{V},\mathcal{E})$, let $l$ be the number of simple loops and denote their lengths by $n_q, q = 1,2,\dots,l$. Let $\mathcal{V}_q$ be the set of vertices traversed by loop $q$. 

Consider a directed cactus graph $\mathbb{G}$ with $N$ vertices, where each vertex $i$ is associated with a transfer matrix $P_i \in \mathcal{RH}_\infty^{m\times m}$. 
The system interconnections are described by a signed adjacency matrix $A = (a_{ij}) \in \mathbb{R}^{N\times N}$, where the entry $a_{ij} \in \{1, -1\}$ indicates a directed edge from vertex $j$ to vertex $i$, and $a_{ij} = 0$ otherwise. To ensure that every simple loop in $\mathbb{G}$ operates as a negative feedback configuration, exactly one edge within each simple loop carries a sign of $-1$, while all other edges in the loop carry a sign of $1$. 
Denote $\widetilde{P} = \mathrm{diag} \{ P_1, P_2,\dots,P_N \} \in \mathcal{RH}_\infty^{Nm\times Nm}$, then the interconnected network is stable if \cite{Zhou1996Robust} 
\begin{equation}              
\label{eq:network_stability_condition_direct_inverse} 
    \left(I + (-A \otimes I) \widetilde{P} \right)^{-1} \in \mathcal{RH}_{\infty}^{Nm\times Nm}. 
\end{equation} 
Directly evaluating this global inverse is computationally intractable for large-scale networks and practically infeasible when exact plant models are uncertain. In contrast, for cactus networks, this condition can be efficiently verified by decomposing the network into basic parallel, cascade and feedback interconnections among the individual subsystems. Consequently, the global stability test reduces to sequentially verifying a series of local loops and their aggregated equivalents.  
This decomposition not only significantly reduces the computational complexity but also provides clear structural insights into the network.

\subsection{Sunflower Case} 
Among cactus graphs, the sunflower graph is a representative structure. 
A cactus graph is called a sunflower graph if one of its loops intersects all other loops.  
Such a special loop is called a central (or global) loop while the other loops are called petal (or local) loops. Furthermore, the pairwise intersections, referred to as hubs, are strictly contained in the central loop.   
An illustration of a sunflower graph is shown in Fig.~\ref{fig:sunflower}.

Without loss of generality, let loop $1$ be the central loop and $\mathcal{V}_1$ denote its vertex set. Let $\mathcal{V}_\text{hub} \subset \mathcal{V}_1$ be the set of hubs shared with petal loops. For each hub $i \in \mathcal{V}_\text{hub}$, let $\mathcal{Q}_i$ denote the index set of the petal loops containing vertex $i$.

To derive stability conditions for such interconnected structures using the $\theta$-symmetric SRG framework, we first introduce the notions of $\theta$-chord hull and $\theta$-arc hull, respectively. For a region $\mathcal{R} \subset \mathbb{C}$ and $\theta \in \mathbb{R}$, define 
\begin{align*} 
        \mathcal{H}_{\theta}^\text{chord} (\mathcal{R}) & = \{ w \in \mathbb{C}: \exists z \in \mathcal{R} \text{ such that } w \in  \mathrm{Chord}_\theta(z) \}, \\ 
        \mathcal{H}_{\theta}^\text{arc} (\mathcal{R})  & = \{ w \in \mathbb{C}: \exists z \in \mathcal{R} \text{ such that } w \in \mathrm{Arc}_\theta^+  (z)  \} . 
\end{align*}  
Accordingly, let $\kappa \in \{ \text{chord}, \text{arc} \}$ denote the specified geometric type. We can define the following matrix sets via their $\theta$-symmetric SRGs as 
\begin{equation}            \label{eq:matrix_set_theta_chord_and_arc_hull}
    \mathcal{M}_{\theta}^\kappa (C) \! = \! \{  M \! \in \! \mathbb{C}^{n\times n} \! : \! \mathrm{SRG}_\theta(M) \! \subset \! \mathcal{H}_\theta^\kappa (\mathrm{SRG}_\theta(C)) \} . 
\end{equation}  
The following lemma is a direct consequence of this definition. 

\begin{lemma}                   \label{lem:theta_chord_arc_hull}
    Let a matrix $C \in \mathbb{C}^{n\times n}$ and $\theta \in \mathbb{R}$ be given. Then $\mathcal{M}_{\theta}^\text{chord}(C)$ satisfies the $\theta$-chord property and $\mathcal{M}_{\theta}^\text{arc}(C)$ satisfies the $\theta$-arc property. 
\end{lemma}

\begin{proof}
    See Appendix \ref{app:proofs}. 
\end{proof}

Clearly, $\mathcal{M}_{\theta}^\text{chord}(C)$ and $\mathcal{M}_{\theta}^\text{arc}(C)$ are over-approximations of the matrix $C$. It is worth noting that they are the tightest over-approximating sets subject to the $\theta$-chord and $\theta$-arc properties, respectively.  

We briefly contrast the proposed $\theta$-arc hull with the S-hull~\cite{Lestas2006ScalableDecentralizedStability}, an over-approximation tool for scalable stability analysis of SISO networks defined as 
\begin{equation*}
    \mathcal{H}_S(\mathcal{R}) = \left(  \mathrm{Co} \left( \sqrt{\mathcal{R}}\right) \right)^2 \quad \text{ with } \sqrt{\mathcal{R}} = \{ x: x^2 \in \mathcal{R} \} . 
\end{equation*}  
While both notions serve to over-approximate complex regions, their structural properties and applicability differ significantly. The S-hull is used to bound the numerical range of specifically arranged matrix products and always contains the origin. In contrast, the $\theta$-arc hull is tailored for the $\theta$-symmetric SRG, a set inherently distinct from the numerical range, enabling it to bound the product of arbitrary matrices. 
Their relative conservatism is thus highly problem-dependent.  

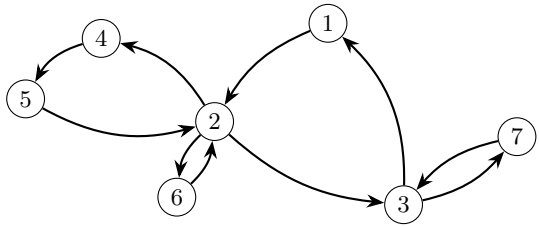
\begin{figure}[htb]
\centering
\begin{tikzpicture}[
        >=Stealth,
        node/.style={circle,draw,minimum size=5mm,inner sep=0pt,font=\small},
        edge/.style={->,thick}
    ]

    \node[node] (1) at (1,2.1) {$1$};
    \node[node] (2) at (-0.5,0.8) {$2$};
    \node[node] (3) at (2,-0.3) {$3$};
    \node[node] (4) at (-2,1.9) {$4$};
    \node[node] (5) at (-3,1.1) {$5$};
    \node[node] (6) at (-1,-0.2) {$6$};
    \node[node] (7) at (3.5,0.6) {$7$};

    \draw[edge] (1) to[bend right=18] (2);
    \draw[edge] (2) to[bend right=18] (3);
    \draw[edge] (3) to[bend right=24] (1);

    \draw[edge] (2) to[bend right=22] (4);
    \draw[edge] (4) to[bend right=22] (5);
    \draw[edge] (5) to[bend right=22] (2);

    \draw[edge] (2) to[bend right=20] (6);
    \draw[edge] (6) to[bend right=20] (2);

    \draw[edge] (3) to[bend right=18] (7);
    \draw[edge] (7) to[bend right=18] (3);

    \end{tikzpicture}
\caption{A sunflower graph with a central loop {\unboldmath$1\to2\to3\to1$} and petal loops attached at hubs {\unboldmath$2$} and {\unboldmath$3$}.}
\label{fig:sunflower} 
\vspace{-0.5em} 
\end{figure}

With these over-approximations in place, we are now ready to characterize the individual subsystems in the sunflower graph. 

\begin{definition}
    For given $ \xi_k(\omega) \in \mathbb{R}, k =1,2,\dots,N,$ and $ \eta_i(\omega) \in \mathbb{R}, i \in \mathcal{V}_1$, matrix set  $\mathcal{P}_i(j\omega),i=1,2,\dots,N$ is said to be an admissible over-approximation of $P_i(j\omega)$ with respect to $(\xi_k(\omega), \eta_i(\omega))$ if 
    \begin{enumerate}[(i)]
    \item petal-loop condition holds: 
    for each petal loop $q \! \in \!  \mathcal{Q}_i$ with $i \! \in \!  \mathcal{V}_{\text{hub}}$, at least $n_q\! - \! 1$ matrix sets in $\{ \mathcal{P}_k(j\omega): k \! \in \! \mathcal{V}_q \}$ satisfy $\xi_k(\omega)$-arc property; 
    \item central-loop condition holds: 
    $\mathcal{P}_i(j\omega)$ satisfies $\eta_i(\omega)$-arc property for all hubs $i \! \in \!  \mathcal{V}_\text{hub}$ and for all but at most one non-hubs $i \! \in \! \mathcal{V}_1\!\setminus\!\mathcal{V}_{\text{hub}}$.  
\end{enumerate} 
\end{definition}
 
Furthermore, to capture the collective effect of the petal loops at each hub $i \in \mathcal{V}_\text{hub}$, define 
\begin{equation*}                           
    \mathcal{W}_i(j\omega) \! = \!\!\!  \sum_{q\in \mathcal{Q}_i } \!      \mathcal{H}_{\theta_i(\omega)}^{\text{chord}} \! \left(  \prod_{k \in \mathcal{V}_q} \! \mathrm{SRG}_{\xi_k(\omega)} ( \mathcal{P}_k(j\omega) ) \! \right)\!, \, i \!\in\! \mathcal{V}_\text{hub} , 
\end{equation*} 
where  $\theta_i(\omega)$ is the identical parameter sum shared by all associated petal loops, satisfying
\begin{equation}                \label{eq:uniform_theta}
    \theta_i(\omega) = \sum_{k\in\mathcal{V}_q} \xi_k(\omega), \qquad \forall  q\in \mathcal{Q}_i   . 
\end{equation}

The following result provides a graphical stability condition for sunflower graphs based on the $\theta$-symmetric SRG framework. This condition is derived hierarchically: first, the stability conditions for the individual petal loops are verified; then, the petal loops are eliminated by absorbing their collective effects into the vertex dynamics on the shared hubs, yielding an equivalent cyclic loop; finally, global stability is determined by verifying this reduced equivalent loop.

\begin{theorem}                     \label{thm:sun_flower_graph_stability}
    Consider a sunflower graph of systems $P_1,P_2,\dots,P_N \in \mathcal{RH}_\infty^{m\times m}$, the graph is stable if there exist functions  $ \xi_k(\omega) \in \mathbb{R}, k =1,2,\dots,N, \eta_i(\omega) \in \mathbb{R}, i \in \mathcal{V}_1$, such that for all $\omega \in [0,\infty]$, 
    \begin{align}
         [-\infty, -1] & \cap  \mathcal{W}_i(j\omega) = \emptyset ,  
         \qquad i \! \in \! \mathcal{V}_\text{hub} , \label{eq:stability_condition_for_petal_loops} \\  
         [-\infty, -1] & \cap \!\!\!  \prod_{i \in \mathcal{V}_1\!\setminus\!  \mathcal{V}_{\text{hub}}} \!\!\!\!\!  \mathrm{SRG}_{\eta_i(\omega)} (\mathcal{P}_i(j\omega) ) \!\! \prod_{i \in \mathcal{V}_{\text{hub}}} \widetilde{\mathcal{W}}_i(j\omega) = \emptyset,    \label{eq:stability_condition_for_central_loop}  
    \end{align} 
    where $\mathcal{P}_i(j\omega) $ is an admissible over-approximation of $P_i(j\omega)$ with respect to $(\xi_k(\omega), \eta_i(\omega))$ and 
    \begin{equation*} 
    \widetilde{\mathcal{W}}_i(j\omega) =   \mathcal{H}_{\eta_i(\omega)-\theta_i(\omega)}^\text{arc} \! \left(  \frac{\mathrm{SRG}_{\eta_i(\omega)}(\mathcal{P}_i(j\omega))}{[1, e^{2j\theta_i(\omega)}] + \mathcal{W}_i(j\omega)} \right) . 
    \end{equation*} 
     
\end{theorem}
 
\begin{proof}
    See Appendix \ref{app:proofs}. 
\end{proof}

\begin{remark}              \label{rem:reduce_to_one_loop}
    For each $i \in \mathcal{V}_\text{hub}$, the $\theta$-chord hull over-approximation in \eqref{eq:stability_condition_for_petal_loops} is required for all but at most one term in the summation. 
\end{remark}

When the central loop of the sunflower graph degenerates to a single vertex, so that all simple loops intersect at one common vertex, the following result holds. Its proof follows directly from Theorem \ref{thm:sun_flower_graph_stability} and is omitted. 

\begin{corollary}           \label{coro:graph_with_one_common_vertex} 
    Consider a sunflower graph of systems $P_1,P_2,\dots,P_N\in\mathcal{RH}_\infty^{m\times m}$ where all simple loops intersect at one common vertex, the graph is stable if there exist functions $\theta(\omega)\in\mathbb{R}$ and $\xi_k(\omega) \in \mathbb{R},k=1,2,\dots,N$, such that for all $\omega \in [0,\infty]$, 
    \begin{align*} 
        [-\infty, -1] &  \cap  \sum_{q = 1}^l \mathcal{H}_{\theta(\omega)}^{\text{chord}} \left(  \prod_{k \in \mathcal{V}_q} \mathrm{SRG}_{\xi_k(\omega)}(\mathcal{P}_k(j\omega)) \right) =\emptyset, \\  
        \theta(\omega) &   =  \sum_{k \in \mathcal{V}_q} \xi_k(\omega), \quad \forall \, q = 1,2,\dots,l, 
    \end{align*} 
    where $\mathcal{P}_k(j\omega)$ is an arbitrary matrix set containing $P_k(j\omega)$, and for each loop $q$, at least $n_q-1$ matrix sets in $\{ \mathcal{P}_k(j\omega): k \in \mathcal{V}_q\}$ satisfy $\xi_k(\omega)$-arc property. 
\end{corollary}

\begin{remark}
    For the case in Corollary \ref{coro:graph_with_one_common_vertex}, an alternative stability condition can be derived by reinterpreting the graph structure: any loop can be designated as the central loop, with all other loops viewed as petal loops intersecting it at one common vertex. The stability condition then follows directly from Theorem \ref{thm:sun_flower_graph_stability}. 
\end{remark}

When there is only one single loop in the graph, as noted in Remark \ref{rem:reduce_to_one_loop}, the $\theta(\omega)$-chord hull over-approximation is no longer needed and Corollary \ref{coro:graph_with_one_common_vertex} reduces to Theorem \ref{thm:cyclic_system_theta_SRG}.

\begin{figure}[htb] 
\vspace{-0.5em} 
\centering  
\subfigure[]{
\begin{minipage}[t]{0.46\linewidth}
\centering 
\begin{tikzpicture}[scale=0.7] 
\useasboundingbox (-1.2,-2.5) rectangle (4.0,2.5);

\begin{scope}   
\draw [white, fill = gray, opacity = 0.5] 
    (15:3) -- (15:0.75) arc (15:330:0.75) -- (330:3) arc (330:375:3) -- cycle;  
\draw (15:3.3) -- (15:0.75) arc (15:330:0.75) -- (330:3.3);
\end{scope} 

\draw[-latex, thick](-1.3,0)--(3.8,0); 
\draw[-latex, thick](0,-2.2)--(0,2.2); 
\node[below] at (3.6,0){$\mathrm{Re}$}; 
\node[left] at (0,2.0){$\mathrm{Im}$};  
\node[below left] at (0,0){$0$}; 
\end{tikzpicture} 
\label{subfig:mixed_gain_phase_small} 
\end{minipage} 
} 
\subfigure[]{
\begin{minipage}[t]{0.46\linewidth}
\centering 
\begin{tikzpicture}[scale=0.7] 
\useasboundingbox (-2.5,-2.5) rectangle (2.5,2.5);

\begin{scope}   
\draw [white, fill = gray, opacity = 0.5] (0,0) -- (-147:1.6) arc (-147:163:1.6) -- cycle;  
\draw (0,0) -- (-147:1.6);
\draw (0,0) -- (163:1.6);
\draw (-147:1.6) arc (-147:163:1.6);
\end{scope} 

\draw[-latex, thick](-2.2,0)--(2.2,0); 
\draw[-latex, thick](0,-2.2)--(0,2.2); 
\node[below] at (2.0,0){$\mathrm{Re}$}; 
\node[left] at (0,2.0){$\mathrm{Im}$};  
\node[below right] at (0,0){$0$}; 
\end{tikzpicture} 
\label{subfig:mixed_gain_phase_large}
\end{minipage} 
} 
\vspace{-0.5em}  
\caption{Illustrations of regions (a) {\unboldmath$\Omega$} and (b) {\unboldmath$\widetilde{\Omega}$}.}  
\label{fig:stability_mixed_gain_phase_region} 
\vspace{-0.5em} 
\end{figure}
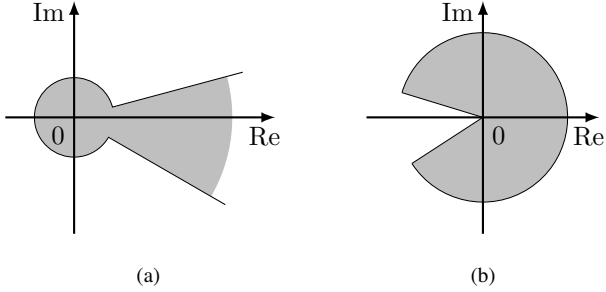

\begin{example}
    Consider the sunflower graph shown in Fig. \ref{fig:sunflower}, where the subsystem $P_1(s)$ is uncertain and the remaining subsystems are given by  
    \begin{align*}
        P_2(s) & \!= \! \frac{
        \begin{bmatrix}
            3s^3 \! + \! 21s^2 \! + \! 46s \! + \! 29 & s^3 \! + \! 3s^2 \! - \! 1 \\ 
            s^3 \! + \! 3s^2 \! - \! 1 &  3s^3 \! + \! 21s^2 \! + \! 46s \! + \! 29 
        \end{bmatrix}
        }{3(s\!+\!5)(2s^2\!+\!8s\!+\!7)},  \\ 
        P_3(s) & \! = \! \frac{ 
        \begin{bmatrix}
            3s \! + \! 2 & -s \\ -s & 3s\! + \! 2 
        \end{bmatrix}
        }{2s\! + \! 1}  , \quad 
        P_4 (s) \! = \! \frac{
        \begin{bmatrix}
            2s \! + \! 7 & 3 \\ 3 & 2s\! + \! 7 
        \end{bmatrix}
        }{2(s\! + \! 2)}, \\ 
        P_5(s) & \! = \!  \frac{
        \begin{bmatrix}
            2s \! + \! 5 & -1 \\ -1 & 2s\! + \! 5 
        \end{bmatrix}
        }{2(s\! + \! 2)}, \quad 
        P_7(s) \! = \!  \frac{
        \begin{bmatrix}
            3s\!+\!2 &  s \\ s & 3s\!+\!2  
        \end{bmatrix}
        }{2s+2}, \\ 
        P_6(s)  & \! = \! \frac{ 
        \begin{bmatrix}
            2s^2\!+\!10s\!+\!9 & 2s\!+\!1 \\ 2s\!+\!1 & 2s^2 \!+\! 10s\!+\!9
        \end{bmatrix} 
        }{2(s\!+\!1)(s\!+\!2)}  . 
    \end{align*} 
    For $P_1(s)$, we have a rough bound, i.e., 
    \begin{equation*}
        \mathrm{SRG}_{\eta_1} \! (P_1(j\omega)) \! \subset\! \Omega :=\! \{ re^{j\theta}\!:\! |r|\! \leq\! 0.6 \text{ or } \!-\!15^\circ \! \leq\! \theta \!\leq \! 10^\circ \} , 
    \end{equation*} for all $\omega \in [0,\infty]$ with $\eta_1 = -2.5^\circ$. An illustration of $\Omega$ is shown in  Fig.~\ref{subfig:mixed_gain_phase_small}. Under this case, the direct computation of the transfer matrix in \eqref{eq:network_stability_condition_direct_inverse} is inapplicable due to the presence of uncertainties. Nevertheless, the stability criterion based on the $\theta$-symmetric SRG in Theorem \ref{thm:sun_flower_graph_stability} remains applicable. 
   Let $\eta_2(\omega) = \xi_2(\omega)$, $\eta_3(\omega) = \xi_3(\omega)$, and $\xi_6(\omega) = \xi_4(\omega)+\xi_5(\omega)$, where the parameters $\xi_k(\omega)$ for $k \in \{2,3,4,5,7\}$ are specifically chosen to generate their corresponding frequency-wise canonical SRGs. An illustration of $\xi_k(\omega)$ is shown in Fig.~\ref{fig:different_theta_choices_for_stability}. Let $\mathcal{P}_k(j\omega) \! = \! \mathcal{M}_{\xi_k(\omega)}^{\text{arc}} (P_k(j\omega)) ), k \!=\! 2,3,\dots,7$, then we have 
    \begin{align*}
        \mathcal{W}_2(j\omega) \! \subset \!  \mathcal{S}[-63^\circ\!, 58^\circ\!, 4.3 ], \ \ 
        \mathcal{W}_3(j\omega) \! \subset\! \mathcal{S}[-50^\circ\!,47^\circ\!,4.2], 
    \end{align*} 
    implying \eqref{eq:stability_condition_for_petal_loops} holds. Moreover, by aggregating the effect of petal loops into the central loop, we have 
    \begin{equation*}
        \prod_{i \in \{2,3\} } \widetilde{\mathcal{W}}_i(j\omega) \subset \widetilde{\Omega} :=   \mathcal{S}[-164^\circ,168^\circ,1.47]. 
    \end{equation*} 
    An illustration of $\widetilde{\Omega}$ is shown in Fig.~\ref{subfig:mixed_gain_phase_large}. It is clear that $[-\infty,-1] \cap \Omega \times \widetilde{\Omega} =\emptyset$, implying \eqref{eq:stability_condition_for_central_loop} holds. Hence it follows from Theorem \ref{thm:sun_flower_graph_stability} that the sunflower graph is stable.  
\end{example}

\begin{figure}
    \centering
    \includegraphics[width=0.6\linewidth]{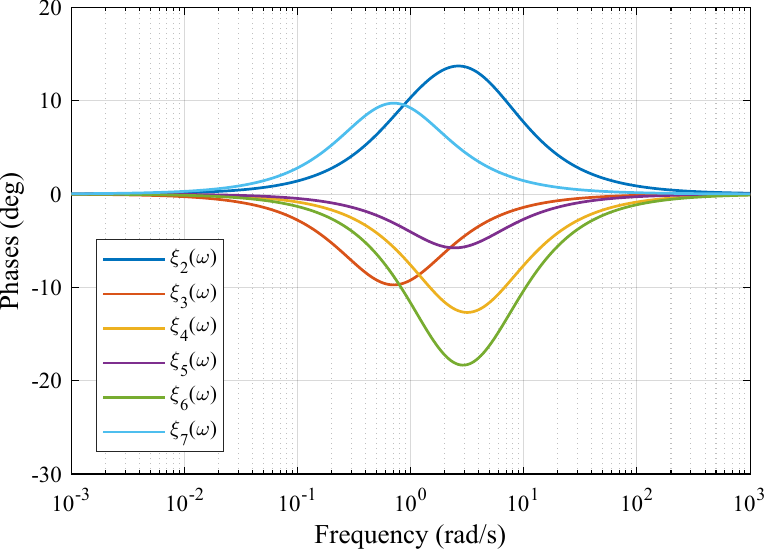}
    \caption{An illustration of the choices of {\unboldmath$\xi_k(\omega), k =2,3,\dots,7$.}}
    \label{fig:different_theta_choices_for_stability} 
    \vspace{-1em} 
\end{figure}

When the gain and phase are considered independently, Theorem \ref{thm:sun_flower_graph_stability} reduces to the following 
decoupled stability criteria for sunflower graphs.  

\begin{corollary}                   \label{coro:sunflower_gain_condition}
    A sunflower graph of systems $P_1,P_2,\dots,$ $P_N \in \mathcal{RH}_\infty^{m\times m}$ is stable if for all $\omega \in [0,\infty]$, 
    \begin{align*} 
        \gamma_i (j\omega) \! > \! 0 , \ \  i \in \mathcal{V}_{\text{hub}},  \text{ and }  \! \prod_{i \in \mathcal{V}_1} \!  \overline{\sigma}(P_i(j\omega)) \!  <  \!  \prod_{i \in \mathcal{V}_{\text{hub}}} \! \gamma_i(j\omega) ,  
    \end{align*} 
    where $\gamma_i(j\omega) \! = \! 1\! -\! \sum_{q \in \mathcal{Q}_i} \! \prod_{k \in \mathcal{V}_q } \! \overline{\sigma}(P_k(j\omega)) $ with $i \in \mathcal{V}_{\text{hub}}$. 
\end{corollary}

The condition $\gamma_i(j\omega) >0$ in Corollary \ref{coro:sunflower_gain_condition} guarantees the local stability of petal loops attached to hub $i$. Aggregating the dynamics of these petal loops tightens the global stability threshold for the central loop. Instead of the single-loop bound of $1$ in Proposition \ref{prop:cyclic_system_mixed_gain_phase}, the allowable gain product is strictly reduced to $\prod_{i\in\mathcal{V}_{\text{hub}}} \!  \gamma_i(j\omega)$. Similarly, this aggregation mechanism also tightens the phase bound for the central loop, yielding the following phase-based criterion. 

\begin{corollary}                   \label{coro:sunflower_phase_condition}
    A sunflower graph of systems $P_1,P_2,\dots,$ $P_N \in \mathcal{RH}_\infty^{m\times m}$ is stable if there exist functions  $ \xi_k(\omega) \in \mathbb{R}, k =1,2,\dots,N, \eta_i(\omega) \in \mathbb{R}, i \in \mathcal{V}_1$, such that for all $\omega \in [0,\infty]$,  
    \begin{align*}  
        & \overline{\phi}_i(j\omega) - \underline{\phi}_i(j\omega) <\pi \text{ and } \phi_i(j\omega) \subset (-\pi,\pi) , \quad  i \! \in \! \mathcal{V}_{\text{hub}},  \\ 
        & \sum_{i \in \mathcal{V}_1} \Psi_{\eta_i(\omega)}(P_i(j\omega)) - \sum_{i \in \mathcal{V}_{\text{hub}}} \Phi_i(j\omega) \subset (-\pi,\pi)  , 
    \end{align*} 
    where 
    \begin{align*} 
        \Phi_i(j\omega)  \! =  &   \Big[ \! \min ( \psi_{\theta_i(\omega)}(I), \phi_i(j\omega) ), \max ( \psi_{\theta_i(\omega)}(I), \phi_i(j\omega) ) \! \Big] , 
        \\  
        \phi_i(j\omega)    =   & [\underline{\phi}_i(j\omega), \overline{\phi}_i(j\omega)] \\ = \Big[ \min_{q \in \mathcal{Q}_i} & \,  \sum_{k \in \mathcal{V}_q}  \underline{\psi}_{\xi_k(\omega)} (P_k(j\omega)),  
         \max_{q \in \mathcal{Q}_i}  \, \sum_{k \in \mathcal{V}_q} \overline{\psi}_{\xi_k(\omega)} (P_k(j\omega))  \Big] , 
    \end{align*} 
    with $\theta_i(\omega)$ given in~\eqref{eq:uniform_theta}. 
\end{corollary}

\subsection{General Cactus Case} 
The stability analysis of the sunflower graph reveals a basic mechanism: verifying the local stability of petal loops, followed by the aggregation of their collective effect into an equivalent central loop.
This principle naturally extends to general cactus graphs via a recursive decomposition from the periphery inward. Accordingly, Algorithm \ref{alg:cactus} formalizes this hierarchical procedure to assess stability layer by layer. Specifically, the outermost loops are evaluated first, and the resulting equivalent systems are propagated inward until the central loop is reached.

\begin{algorithm}[htb]
\caption{Stability Verification for Cactus Networks}
\label{alg:cactus}
\begin{algorithmic}[1]
\REQUIRE Cactus graph $\mathbb{G}$ 
\ENSURE Stability conclusion. 

\STATE Initialize $\mathbb{G}_0 \leftarrow \mathbb{G}$. Select a central loop as loop $1$.   
\STATE Construct a directed spanning tree $\mathbb{T}_0$ by taking loop $1$ as the root and other simple loops as nodes. 
\STATE $k \leftarrow 0$. 

\REPEAT 
    \STATE Identify the leaf nodes $\mathcal{L}_k$ and the corresponding parent nodes $\mathcal{P}_k$ in $\mathbb{T}_k$.  
    \STATE Identify the set of shared vertices in $\mathbb{G}_k$ as  
    \begin{equation*}
        \mathcal{I}_k = \left\{ \mathcal{V}_p \cap  \mathcal{V}_q: p \in \mathcal{L}_k, q \in \mathcal{P}_k  \right\} .  
    \end{equation*} 

    \vspace{-0.4em} 
    \STATE For each $i \in \mathcal{I}_k$, verify the stability of the attached leaf loops via conditions analogous to \eqref{eq:stability_condition_for_petal_loops}. 
    \STATE Absorb $\mathcal{L}_k$ into $\mathcal{I}_k$ as aggregated equivalent feedback blocks, yielding the reduced graph $\mathbb{G}_{k+1}$. 
    \STATE Update $\mathbb{T}_{k+1} \leftarrow \mathbb{T}_k \! \setminus \! \mathcal{L}_k$. 
    \STATE $k \leftarrow k+1$. 

\UNTIL $\mathbb{T}_{k} = (\{ 1 \},\emptyset )$. 

    \STATE Verify the stability of the single loop graph $\mathbb{G}_{k}$ via conditions analogous to  \eqref{eq:stability_condition_for_central_loop}.  
    
\RETURN Stable if all verification steps pass. 
\end{algorithmic}
\end{algorithm}

\begin{remark} 
    The hierarchical decomposition principle underlying the cactus networks naturally extends to broader topologies. 
    This scalability renders the $\theta$-symmetric SRG a powerful analytical tool for complex systems, including leader–follower architectures and networks partitioned into multiple strongly connected components.  
\end{remark}

\section{Conclusions} 
\label{section:conclusion}

In this paper, we systematically develop the notion of $\theta$-symmetric SRG and apply it to the stability analysis of cactus dynamic networks. In particular, we establish a connection between the phase of $\theta$-symmetric SRG and a norm minimization problem, thereby enabling the computation of $\theta$-segmental phase via SDP.  
Then we derive the submultiplicative and subadditive properties of $\theta$-symmetric SRG, which facilitates the study of nonsingularity of product‑type and sum‑type return difference matrices. Furthermore, we derive necessary and sufficient conditions for the robust stability of cyclic interconnections, the fundamental building blocks of cactus networks. Building upon these foundational results, we construct a layer-by-layer stability framework for general cactus dynamic networks. 

The $\theta$-symmetric SRG framework refines the phase information, enabling us to capture the lead and lag properties of the system and serving as a more appropriate MIMO extension of the classical Nyquist plot. Moreover, as demonstrated in cactus networks, the proposed framework exhibits elegant compositional properties, thereby offering a promising method to systematically analyze more complex networks.
Future work includes applying the $\theta $-symmetric SRG to controller design and generalizing the framework to nonlinear dynamic systems. 

\appendices

\section{Proofs of Some Technical Results} 
\label{app:proofs}

\proofof{Proof of Proposition \ref{prop:theta_segmental_phase_LMI_graphical_interpretation}} 
    We first show (i) $\Rightarrow$ (ii). Suppose (i) holds, then 
    \begin{equation*}
        \arccos \frac{\mathrm{Re}\{ x^He^{-j\theta}Cx \}}{\|Cx\|\|x\|} \leq \delta , \quad \forall 0 \neq x \in \mathbb{C}^n, Cx \neq 0 , 
    \end{equation*} implying 
    $ 
    \mathrm{Re}\{x^He^{-j\theta}Cx \} \geq \cos\delta   \|Cx\|\|x\| . 
    $ 
    Since both sides of this inequality are nonnegative, squaring both sided yields 
    \begin{equation}                \label{eq:segmental_phase_to_NR}
        \left( \frac{\mathrm{Re} \{ x^He^{-j\theta}Cx \}}{\|x\|^2}  \right) ^2 \geq \cos^2\delta  \frac{\|Cx\|^2}{\|x\|^2} . 
    \end{equation} 
    Since $T = \mathrm{Re}\{ e^{-j\theta} C\} + jC^HC$, it follows from \eqref{eq:segmental_phase_to_NR} that  
    \begin{equation*}
        W(T) \subset \{ a+bj\in\mathbb{C}: a^2 \geq b \cos^2\delta , a,b \in \mathbb{R} \}, 
    \end{equation*} 
    which, in view of the convexity of $W(T)$, is equivalent to the existence of $\mu>0$ such that 
    \begin{equation*}
        W(T) \subset \{ a+bj \in \mathbb{C}:  2\mu a \geq \cos\delta (b+\mu^2), a , b \in \mathbb{R} \} . 
    \end{equation*} 
    This implies (ii) holds. 
    
    Next we show (ii) $\Rightarrow$ (iii). Suppose (ii) holds, then there exists $\mu>0$ such that   
    \begin{equation}                 \label{eq:LMI_to_SRG_proof}  
        \cos\delta(\|Cx\|^2  \! + \! \mu^2 \|x\|^2) \! \leq \! 2 \mu x^H \mathrm{Re}\{e^{-j\theta} C \} x , \ \ \ \forall  0 \! \neq \! x \!\in\! \mathbb{C}^n ,  
    \end{equation} implying
    \begin{equation*}
        \cos\delta (C^HC+ \mu^2I) \leq \mu ( e^{-j\theta} C+ e^{j\theta}C^H) .
    \end{equation*} By the Schur complement theorem, this is equivalent to the LMI \eqref{eq:LMI_theta_segmental_phase}, yielding that (iii) holds. 
    
    Then we show (iii) $\Rightarrow$ (iv). Suppose (iii) holds, then \eqref{eq:LMI_to_SRG_proof} holds. For any $z \in \mathrm{SRG}_\theta(C)$, there holds $\{ z, \overline{z}e^{2j\theta} \} \subset \mathrm{SRG}_\theta(C)$. Without loss of generality, we can denote $z$ as $z = e^{j\theta}(a+bj)$ where 
    \begin{equation}            \label{eq:a_b_parameter}
        a = \frac{\mathrm{Re}\{ x^He^{-j\theta} Cx\}}{\|x\|^2}, \quad  
        b =  \sqrt{\frac{\|Cx\|^2}{\|x\|^2} - a^2 }, 
    \end{equation}  
    for some nonzero $x \in \mathbb{C}^n$. In view of \eqref{eq:LMI_to_SRG_proof}, we have  
    \begin{align*}
        \left(a - \frac{\mu}{\cos\delta}\right)^2 + b^2 & = -2a \frac{\mu}{\cos\delta} + \frac{\mu^2}{\cos^2\delta} + \frac{\|Cx\|^2}{\|x\|^2} \\ & \leq -\frac{\|Cx\|^2}{\|x\|^2} - \mu^2 + \frac{\mu^2}{\cos^2\delta} + \frac{\|Cx\|^2}{\|x\|^2} \\ 
        & = \mu^2 \tan^2\delta  , 
    \end{align*} 
    implying (iv) holds. 
    
    Finally we show (iv) $\Rightarrow$ (i). Suppose (iv) holds, then $\mathrm{SRG}_\theta(C)$ is contained in the following cone 
    \begin{equation*} 
        \mathrm{SRG}_\theta(C) \subset \mathcal{D}[\theta,\mu] \subset \{ z \in \mathbb{C}: \angle z \in [\theta-\delta,\theta+\delta] \} . 
    \end{equation*} 
    Since $\Psi_\theta(C)$ corresponds to the minimal cone containing $\mathrm{SRG}_\theta(C)$, we have $\Psi_\theta(C) \subset \theta + [-\delta, \delta]$, implying (i) holds. This completes the proof.  
\endproof

\vspace{0.5em}

\proofof{Proof of Theorem \ref{thm:larger_than_pi_spread}}
    The equivalence (i) $\Leftrightarrow$ (ii) follows similarly as Proposition \ref{prop:theta_segmental_phase_LMI_graphical_interpretation}. Note that (i) holds if and only if 
    \begin{equation*}
        W(T) \! \subset \! \{ a \! + \! bj \! \in \! \mathbb{C}: a > \sqrt{b} \cos\delta , b > 0, a,b\in\mathbb{R} \}, 
    \end{equation*} 
    where $T = \mathrm{Re}\{ e^{-j\theta} C\} + j C^HC $. In view of the convexity of $W(T)$, this is equivalent to that for any $\mu>0$, there holds 
    \begin{equation*}
        W(T) \! \subset \! \{ a \! + \! bj\in\mathbb{C}: 2\mu a \! > \! \cos\delta (b\! + \! \mu^2),b\! > \! 0,a,b\! \in \! \mathbb{R}\} . 
    \end{equation*} 
    This can be rewritten as (ii). 

    Then we show (ii) $\Rightarrow$ (iii). Suppose (ii) holds, then for any $\mu>0$, we have 
    \begin{equation*}
        Q(\mu) : = \mu^2 \cos\delta \, I  - 2\mu \mathrm{Re} \{ e^{-j\theta} C \} + \cos\delta \, C^HC < 0 . 
    \end{equation*}  
    Since $\cos\delta < 0$ and $0\notin \mathrm{SRG}_\theta(C)$, it follows that $Q(0) < 0, Q(\infty) < 0$. Due to the continuity of the eigenvalues of $Q(\mu)$ on parameter $\mu$, $\lambda(Q(\mu))$ will not touch zero. Consequently, the quadratic eigenvalue problem $ Q(\mu) x = 0$ has no eigenvalues in $(0,\infty)$ \cite{Tisseur2001QuadraticEigenvalueProblem}. Denote $y = \mu x$, then $ Q(\mu) x = 0$ can be rewritten as a GEP 
    \begin{equation*}
        M \begin{bmatrix}
            x \\ y 
        \end{bmatrix} = \mu N \begin{bmatrix}
            x \\ y 
        \end{bmatrix} , 
    \end{equation*} 
    implying (iii) holds.  

    The implication (iii) $\Rightarrow$ (ii) can be established by reversing the preceding arguments. This completes the proof. 
\endproof

\vspace{0.5em} 

\proofof{Proof of Theorem \ref{thm:matrix_iff_condition}}  
    For brevity, we only detail the proof for $\mathcal{M}[\alpha_i,\beta_i,\gamma_i]$, as the case for $\mathcal{M}_c[\alpha_i,\beta_i,\gamma_i]$ follows along similar lines and is omitted. 
    We first show the sufficiency. For any $A_i\in\mathcal{M}[\alpha_i,\beta_i,\gamma_i]$, there exists $\theta_i \in \mathbb{R}$ such that 
    \begin{equation*}
        \mathrm{SRG}_{\theta_i} (A_i) \subset \mathcal{R}[\alpha_i,\beta_i,\gamma_i], \quad i \in \mathcal{K}'. 
    \end{equation*} 
    For each $\theta_i$, consider the matrix set $\mathcal{M}_{\theta_i}^\text{arc}(A_i)$ defined in \eqref{eq:matrix_set_theta_chord_and_arc_hull}. In view of Lemma \ref{lem:theta_chord_arc_hull}, we have 
    \begin{equation*}
        \mathrm{SRG}_{\theta_i} (\mathcal{M}_{\theta_i}^\text{arc}(A_i))   = \mathcal{H}_\theta^{\text{arc}} (\mathrm{SRG}_{\theta_i}(A_i))  
          \subset \mathcal{R}[\alpha_i,\beta_i,\gamma_i]  . 
    \end{equation*} 
    Hence, \eqref{eq:matrix_nonsingularity_iff_condition} implies that there exist $\theta$ and $ \theta_i,i \in \mathcal{K}'$, such that 
    \begin{equation*}
        -1 \notin \mathrm{SRG}_{\theta}(A_k) \prod_{i \in \mathcal{K}'} \mathrm{SRG}_{\theta_i} (\mathcal{M}_{\theta_i}^\text{arc} (A_i)) . 
    \end{equation*} 
    Since $\mathcal{M}_{\theta_i}^\text{arc}(A_i)$ satisfies the $\theta_i$-arc property, it follows from Theorem \ref{thm:matrix_nonsingular_condition} that $I+\prod_{i=1}^NA_i$ is nonsingular. 
    
    Then we show the necessity. Suppose by contradiction that for any $\theta\in\mathbb{R}$, there holds 
    \begin{equation*}
        -1 \in \mathrm{SRG}_\theta(A_k) \prod_{i \in \mathcal{K}'} \mathcal{R}[\alpha_i,\beta_i,\gamma_i] . 
    \end{equation*} 
    Then for a particular choice $\theta = - \frac{1}{2}\sum_{i \in \mathcal{K}'} (\beta_i+\alpha_i)$, there exist nonzero $x \in \mathbb{C}^n$ and $z_i\in\mathcal{R}[\alpha_i,\beta_i,\gamma_i], i \in \mathcal{K}'$, such that 
    \begin{equation*}
        -1 = \frac{\|A_kx\|}{\|x\|} \exp \{ j (\theta + \angle_\theta(x,A_kx))\} \prod_{i \in \mathcal{K}'} z_i . 
    \end{equation*} 
    Treating the gain and phase separately, we have 
    \begin{align*}
        \prod_{i \in \mathcal{K}'} |z_i| = \frac{\|x\|}{\|A_kx\|},   \quad 
        \sum_{i \in \mathcal{K}'} \angle z_i   = \pi - \theta  - \angle_{\theta}(x,A_kx) .  
    \end{align*}
    Denote 
    \begin{align*} 
        y & =   \begin{bmatrix}
            \sqrt{\frac{b+a}{2b}} & \sqrt{\frac{b-a}{2b}} & 0 
        \end{bmatrix} ^T \in \mathbb{R}^{n} , \\  
        \tilde{A}_i &= \begin{bmatrix}
            z_i & 0 & 0 \\ 0 &  e^{j(\beta_i+\alpha_i) } \overline{z}_i & 0 \\ 0 & 0 & e^{\frac{1}{2}j(\beta_i+\alpha_i)} |z_i|I 
        \end{bmatrix}  , \quad i \in \mathcal{K}' ,  
    \end{align*} 
    where 
    \begin{align*}
        a & = \mathrm{Im}\{ e^{-j\theta}x^HA_kx \},   \\  b & = \sqrt{\|A_kx\|^2\|x\|^2 - \left( \mathrm{Re}\{ e^{-j\theta}  x^HA_kx \} \right)^2} . 
    \end{align*}
    Then one can verify that 
    \begin{equation*}
        y^H \prod_{i \in \mathcal{K}'}\tilde{A}_i y  =  - \frac{x^HA_k^Hx}{\|A_kx\|^2} , \quad   \left\|\prod_{i \in \mathcal{K}'} \tilde{A}_i y \right\| = \frac{\|x\|}{\|A_k x\|} .  
    \end{equation*} 
    Similar to the proof of \cite[Theorem IV. 1]{Zhang2025DW_Shell}, denote $B = \prod_{i \in \mathcal{K}'} \tilde{A}_i $ and decompose 
    \begin{align*}
        By & \! =\! (y^HBy ) y \! + \! \sqrt{\|By\|^2-|y^HBy|^2} \,  y_\bot , \\ 
        x & \! = \! \frac{x^HA_k^Hx}{\|A_kx\|} \frac{A_kx}{\|A_kx\|} \! + \! \sqrt{\|x\|^2  \! - \! \frac{|x^HA_k^Hx|^2}{\|Ax\|^2}} \left( \! \frac{A_kx}{\|A_kx\|} \! \right)_\bot , 
    \end{align*}  
    where $y_\bot$ and $\left( \frac{A_kx}{\|A_kx\|} \right)_\bot$ are unit vectors orthogonal to $y$ and $ \frac{A_kx}{\|A_kx\|}$, respectively. Let $U$ be a unitary matrix such that 
    \begin{equation*}
        U \begin{bmatrix}
            \dfrac{A_kx}{\|A_kx\|} & - \left( \dfrac{A_kx}{\|A_kx\|}  \right)_\bot 
        \end{bmatrix} = \begin{bmatrix}
            y & y_\bot  
        \end{bmatrix} . 
    \end{equation*} 
    Let 
    $ 
    A_i = U^H \tilde{A}_i U  $ for $ i  \in \mathcal{K}', 
    $ 
    then $A_i \in \mathcal{M}[\alpha_i,\beta_i,\gamma_i]$ and  
    \begin{align*} 
        & (I + A_k A_{k+1} \cdots A_N A_1 \cdots A_{k-1} ) \frac{A_kx}{\|A_kx\|} \\ & \qquad \qquad \qquad  = \left(I + A_k U^H \prod_{i\in\mathcal{K}'} \tilde{A}_i U \right) \frac{A_kx}{\|A_kx\|}  =  0 , 
    \end{align*} 
    implying $I + A_k A_{k+1} \cdots A_N A_1 \cdots A_{k-1} $ is singular. Since a cyclic permutation does not affect the singularity, we have $I+\prod_{i=1}^NA_i$ is singular, which is a contradiction. This completes the proof. 
\endproof

\vspace{0.5em} 

\proofof{Proof of Proposition \ref{prop:interpolation_result}}
    Suppose \eqref{eq:interpolation_condition} holds, we will consider the cases $\beta = \frac{\pi}{2}$ and $\beta \in [0,\frac{\pi}{2})$ separately. 
    
    If $\beta = \frac{\pi}{2}$, for any $z \in \mathrm{SRG}(\Delta)\subset \mathcal{R}[\beta,\gamma]$, there holds $|z|\leq \gamma \cos\angle z$. Denote 
    $
    M_0 = \frac{2}{\gamma} \Delta - I , 
    $ 
    then we have $M_0^HM_0 \leq I$.  
    Based on the small gain interpolation method \cite{Zhou1996Robust,Ball1990Interpolation}, there exists a transfer matrix $M \in \mathcal{RH}_\infty^{m\times m}$ such that $\|M\|_\infty \leq 1$ and $M(j\omega_0) = M_0$. Let 
    \begin{equation*}
        P(s) = \frac{\gamma}{2} \left( M(s) + I \right)  , 
    \end{equation*} 
    then one can verify that $P \in \mathfrak{P}[\beta,\gamma]$ and $P(j\omega_0) = \Delta$.  
    
    If $\beta \in [0,\frac{\pi}{2})$, it follows from Proposition \ref{prop:theta_segmental_phase_LMI_graphical_interpretation} that \eqref{eq:interpolation_condition} implies the existence of $\mu>0$ such that 
    \begin{equation*}
        \cos\beta (\Delta^H\Delta + \mu^2I ) \leq \mu (\Delta+\Delta^H). 
    \end{equation*} 
    Denote  
    $ 
    N_0 = \frac{1}{\mu \tan\beta} \left( \Delta - \frac{\mu}{\cos\beta}  I \right)  ,  
    $ 
    then we have $N_0^HN_0\leq I$. There also exists a transfer matrix $N \in \mathcal{RH}_\infty^{m\times m}$ such that $\|N\|_\infty \leq 1$ and $N(j\omega_0) = N_0$. Let 
    \begin{equation*}
        P(s) = \mu \tan\beta \, N(s) + \frac{\mu}{\cos\beta} I , 
    \end{equation*} 
    then one can verify that $P \in \mathfrak{P}[\beta,\gamma]$ and $P(j\omega_0) = \Delta$. This completes the proof.  
\endproof

\vspace{0.5em} 

\proofof{Proof of Lemma \ref{lem:theta_chord_arc_hull}} 
    We show that $\mathcal{M}_{\theta}^\text{chord}(C)$ satisfies the $\theta$-chord property. The proof for $\mathcal{M}_{\theta}^\text{arc}(C)$ satisfying the $\theta$-arc property follows by a similar argument. By definition, we have 
    \begin{equation*}
        \mathrm{SRG}_\theta(\mathcal{M}_{\theta}^\text{chord}(C)) \subset \mathcal{H}_\theta^{\text{chord}}(\mathrm{SRG}_\theta(C)). 
    \end{equation*} 
    It remains to show the reverse inclusion. Consider a subset of $\mathcal{M}_{\theta}^\text{chord}(C)$ defined by 
    \begin{equation*}
        \mathcal{M} = \{ zI \in \mathbb{C}^{n\times n}: z \in \mathcal{H}_\theta^{\text{chord}}(\mathrm{SRG}_\theta(C)) \} . 
    \end{equation*}
    Then there holds 
    \begin{equation*}
        \mathrm{SRG}_\theta(\mathcal{M}_{\theta}^\text{chord}(C)) \supset \mathrm{SRG}_\theta(\mathcal{M}) = \mathcal{H}_\theta^{\text{chord}}(\mathrm{SRG}_\theta(C)) . 
    \end{equation*} 
    Hence $\mathcal{M}_{\theta}^\text{chord}(C)$ satisfies the $\theta$-chord property, which completes the proof. 
\endproof

\vspace{0.5em} 

\proofof{Proof of Theorem \ref{thm:sun_flower_graph_stability}}
    For any $i \in \mathcal{V}_{\text{hub}}$, the petal loops containing vertex $i$ form a parallel interconnection of multiple cascade loops due to the sunflower structure. Hence, this interconnection is stable if and only if 
    \begin{equation}                \label{eq:petal_loop_stability_defn}
        \left( I + \sum\nolimits_{q \in \mathcal{Q}_i}  L_q(s) \right)^{-1} \in \mathcal{RH}_\infty^{m\times m}, \quad i \in \mathcal{V}_{\text{hub}}, 
    \end{equation} 
    where $L_q(s)$ is the loop transfer matrix of loop $q$. Then we will show that \eqref{eq:stability_condition_for_petal_loops} implies \eqref{eq:petal_loop_stability_defn}. Due to the submultiplicative property in Theorem \ref{thm:sum_product_properties},  \eqref{eq:stability_condition_for_petal_loops} implies for all $i \!\in\! \mathcal{V}_\text{hub}$, 
    \begin{equation*}
        [-\infty, -1] \cap  \sum\nolimits_{q \in \mathcal{Q}_i }   \mathcal{H}_{\theta_i(\omega)}^{\text{chord}}  \left(   \mathrm{SRG}_{\theta_i(\omega)} \left( \mathcal{L}_q(j\omega) \right)  \right) = \emptyset,  
    \end{equation*}  
    where $\mathcal{L}_q(j\omega) = \prod_{k \in \mathcal{V}_q} \mathcal{P}_k(j\omega) $.  
In view of \eqref{eq:matrix_set_theta_chord_and_arc_hull}, 
we have  $\mathcal{L}_q(j\omega) \subset \mathcal{M}_{\theta_i(\omega)}^{\text{chord}}(\mathcal{L}_q(j\omega))$ and  
    \begin{align*}
        & \mathrm{SRG}_{\theta_i(\omega)} \! \left(  \sum\nolimits_{q \in \mathcal{Q}_i}   \mathcal{M}_{\theta_i(\omega)}^{\text{chord}}(\mathcal{L}_q(j\omega)) \! \right)  \\ &   \qquad \qquad \quad   \subset  \sum\nolimits_{q \in \mathcal{Q}_i }  \mathrm{SRG}_{\theta_i(\omega)} \left( \mathcal{M}_{\theta_i(\omega)}^{\text{chord}}(\mathcal{L}_q(j\omega))   \right) \\ 
        & \qquad  \qquad \quad  = \sum\nolimits_{q \in \mathcal{Q}_i}  \mathcal{H}_{\theta_i(\omega)}^{\text{chord}}  \left(   \mathrm{SRG}_{\theta_i(\omega)} \left( \mathcal{L}_q (j\omega) \right)  \right) . 
    \end{align*} 
    Then we have 
    $ [-\infty, -1]  \cap   \mathrm{SRG}_{\theta_i(\omega)}\left( \sum_{q \in \mathcal{Q}_i} \mathcal{L}_q(j\omega) \right) =\emptyset $.  
    yielding that  $\sum_{q \in \mathcal{Q}_i} L_q(j\omega)$ has no eigenvalues in $[-\infty,-1]$. 
    Hence it follows from the generalized Nyquist criterion that \eqref{eq:petal_loop_stability_defn} holds. 
    
    Then we reduce the collective effect of petal loops into the hubs $i \in \mathcal{V}_{\text{hub}}$ on the central loop. Consequently, the sunflower graph is transformed to an equivalent single loop system, with components defined by 
    \begin{equation*}
        \widehat{P}_i(s) \! = \! \left\{ 
                \begin{array}{ll} 
                     P_i(s) \left(  I \!  + \!  \sum_{q\in\mathcal{Q}_i} L_q(s) \right)^{-1}  & i \! \in \!  \mathcal{V}_{\text{hub}} , \\ 
                     P_i(s) & i \! \in \! \mathcal{V}_1 \!\setminus\! \mathcal{V}_{\text{hub}}  . 
                \end{array}
        \right.
    \end{equation*} 
    Since $P_i(s) \in \mathcal{RH}_\infty^{m\times m}$ and \eqref{eq:petal_loop_stability_defn} holds, it follows that $\widehat{P}_i(s) \in \mathcal{RH}_\infty^{m\times m}$. To show the sunflower graph is stable, it suffices to show 
    \begin{equation}                    \label{eq:central_loop_stability_defn}
        \left(I+ \prod\nolimits_{i \in \mathcal{V}_1} \widehat{P}_i(s) \right)^{-1} \in \mathcal{RH}_\infty^{m\times m}, 
    \end{equation} 
    where the product is taken in the order consistent with the connection along the central loop. We will show that \eqref{eq:stability_condition_for_central_loop} implies \eqref{eq:central_loop_stability_defn}. Note that for any $i \in \mathcal{V}_{\text{hub}} $, 
    \begin{align*}
        \mathrm{SRG}_{\hat{\eta}_i(\omega)} ( \widehat{P}_i(j\omega) ) & \! \subset \!  \frac{\mathrm{SRG}_{\eta_i(\omega)}(\mathcal{P}_i(j\omega))}{\mathrm{SRG}_{\theta_i(\omega)} \!  \left( I + \sum_{q \in \mathcal{Q}_i}   \mathcal{L}_q(s) \! \right)  }   \\ 
        & \! \subset \! \frac{\mathrm{SRG}_{\eta_i(\omega)}(\mathcal{P}_i(j\omega))}{  
        [1,e^{2j\theta_i(\omega)}] 
        \! + \!  \mathcal{W}_i(j\omega)   
        } , 
    \end{align*} 
    where $\hat{\eta}_i(\omega) = \eta_i(\omega)-\theta_i(\omega)$.  
    Hence \eqref{eq:stability_condition_for_central_loop} indicates that 
    \begin{align*}
        & [-\infty, -1]  \cap  \Bigg(  \prod_{i \in \mathcal{V}_1\!\setminus\! \mathcal{V}_{\text{hub}}} \!\!\!\!\!\!   \mathrm{SRG}_{\eta_i(\omega)}   \left( P_i(j\omega)  \right) \\ 
        &  \qquad \quad   \times  \prod_{i \in \mathcal{V}_\text{hub}} \!   \mathcal{H}_{\hat{\eta}_i(\omega)}^{\text{arc}}    \left(  \mathrm{SRG}_{\hat{\eta}_i(\omega)}   \left(   \widehat{P}_i(j\omega)  \right)  \right) \Bigg) = \emptyset  , 
    \end{align*}  
    which, in view of the submultiplicative property 
    yields that 
    $\prod_{i \in \mathcal{V}_1} \widehat{P}(j\omega)$ has no eigenvalues in $[-\infty,-1]$.  
Hence the sunflower graph is stable, which completes the proof. 
\endproof

\section*{Acknowledgment}

The authors would like to thank Ding Zhang of Australian National University for valuable discussions. 

\section*{References}

\bibliographystyle{IEEEtran}
\bibliography{TransReference}

\end{document}